\documentclass[journal]{IEEEtran}
\usepackage[T1]{fontenc}
\usepackage{amsmath}
\usepackage{graphicx}
\usepackage{epsfig}
\usepackage{amsfonts}
\usepackage{amssymb}
\usepackage{gensymb}
\usepackage{caption}
\usepackage{epstopdf}
\usepackage{tikz}
\usepackage{color}
\usepackage{array}
\usepackage{tabu}
\usepackage{algorithm}
\usepackage{subcaption}
\usepackage{relsize}
\usepackage{amsthm}
\usepackage{stfloats}

\usepackage{float} 
\usepackage{algpseudocode}
\usepackage[nospace,noadjust]{cite}

\newcommand{\argmin}{\mathop{\text{argmin}}}

\newtheorem{theorem}{Theorem}[section]

\newtheorem{lemma}[theorem]{Lemma}
\begin{document}
\title{{ Constellation Design for Media-based Modulation 
using Block Codes and Squaring Construction}
}

\author{Bharath Shamasundar and A. Chockalingam \\ 
Department of ECE, Indian Institute of Science, Bangalore
\thanks{This work was supported in part by the 
Visvesvaraya PhD Scheme of the Ministry of Electronics \& IT, 
Government of India, J. C. Bose National Fellowship, Department of Science and Technology,
Government of India, and the Intel India Faculty Excellence Program.}}

\maketitle

\begin{abstract}
Efficient constellation design is important for improving {\color{black}
performance}
in communication systems. The problem of multidimensional constellation 
design has been studied extensively in the literature in the context of 
multidimensional coded modulation and space-time coded MIMO systems.
Such constellations are formally called as lattice codes, where a finite 
set of points from a certain high dimensional lattice is chosen based 
on some criteria. In this paper, we consider the problem of 
constellation/signal set design for media-based modulation (MBM), a 
recent {\color{black} MIMO} channel modulation scheme with promising 
theoretical and 
practical benefits. Constellation design for MBM is fundamentally 
different from those for multidimensional coded modulation and 
{\color{black} conventional} 
MIMO systems mainly because of the inherent sparse structure of the MBM 
signal vectors. Specifically, we need a {\em structured sparse lattice 
code} with good distance properties. In this work, we show that using 
an $(N,K)$ non-binary block code in conjunction with the lattice based 
multilevel squaring construction, it is possible to systematically 
construct a signal set for MBM with certain guaranteed minimum distance. 
{\color{black} The MBM signal set obtained using the proposed construction 
is shown to achieve significantly improved bit error performance compared 
to conventional MBM signal set. In particular,} the proposed signal set is 
found to achieve higher diversity slopes in the low-to-moderate SNR regime. 
\end{abstract}
\begin{IEEEkeywords}
Media-based modulation, mirror activation pattern, MAP-index coding, 
squaring construction. 
\end{IEEEkeywords}

\section{Introduction}
\label{sec1}
{\color{black}Media-based modulation (MBM) is a recent MIMO transmission 
technique which uses a single {\color{black} transmit radio frequency (RF)} 
chain and multiple {\color{black} RF radiation elements.} It has compact 
overall structure 
compared to the conventional MIMO systems and achieves superior rate and 
performance \cite{mbm1}{\color{black}-\cite{duman}}.} 
Specifically, MBM uses digitally controlled parasitic elements called 
RF mirrors, which act as signal scatterers in the near field of the 
transmit antenna {\color{black} (see Fig. \ref{figmbm})}. Each of these 
RF mirrors can be in one of the two states, namely, ON or OFF, 
based on the control inputs which depend on the information bits. An 
RF mirror reflects the transmit signal in the ON state, and allows the 
signal to pass through in the OFF state. If there are $m_{rf}$ RF 
mirrors, then there are $N_m\triangleq 2^{m_{rf}}$ different ON/OFF 
combinations, called `mirror activation patterns' (MAP). Each of these 
MAPs creates a different near field geometry for the transmit signals. 
In a rich scattering environment, even a small perturbation in the 
near field will be augmented by random reflections, and hence results 
in a different end-to-end channel. Therefore, in MBM,  using $m_{rf}$ 
RF mirrors, $N_m$ independent channels can be created corresponding to 
$N_m$ different MAPs. These different MAPs can be represented by the 
$N_m$ MAP indices, $\mathcal{M}=\{0,1,\cdots, N_m-1\}$. The transmitter 
activates one of these MAPs (equivalently, selects one of the indices 
from $\mathcal{M}$) based on $m_{rf}$ information bits and transmits a 
symbol from a conventional modulation alphabet $\mathbb{A}$ (say, QAM), 
which conveys $\log_2|\mathbb{A}|$ bits. The achieved rate in MBM is, 
therefore, given by $\eta_{\text{\tiny MBM}}=m_{rf}+\log_2|\mathbb{A}|$ 
bits per channel use (bpcu). It has been shown that MBM can achieve good 
bit error performance in the point-to-point setting compared to 
conventional SIMO{\color{black}/MIMO systems \cite{mbm1}-\cite{mbm4}}.
In \cite{mbm6}, MBM is studied in the context 
of space-time coding and significant performance gains are reported. 
Inspired by the notion of quadrature spatial modulation, quadrature 
channel modulation schemes using MBM are proposed in \cite{mbm7},\cite{mbm8}. 
In \cite{mbm9}, MBM is used for the uplink in massive MIMO systems, and 
the possible gains in terms of reduction in the required number of 
base station receive antennas are highlighted. {\color{black} Recently, 
practical implementation of MBM using reconfigurable metasurfaces to 
alter the near-field radiation characteristics have been proposed in 
\cite{vinoy1},{\color{black} \cite{duman}}.} {\color{black} Our new 
contribution in the present work is} on designing efficient signal sets 
for MBM that can {\color{black} achieve significantly improved system 
performance}.  

\begin{figure}[t]
\centering
\includegraphics[scale=0.8]{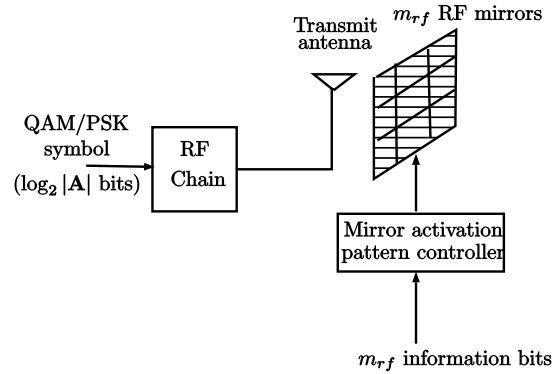}
\caption{Schematic representation of MBM transmitter.}
\vspace{-3mm}
\label{figmbm}
\end{figure}

Constellation design is one of the {\color{black} important means to improve 
performance} in wireless communication systems. In early literature on 
constellation design, several two-dimensional constellations are 
conceived based on different criteria 
\cite{constellation1}-\cite{constellation3}. 
Recognizing the limited SNR gains achievable by two-dimensional 
constellations and the possibility of higher SNR efficiency by going 
for higher dimensions, efficient multidimensional constellations are 
proposed in the literature \cite{constellation4}-\cite{constellation6}. 
The general framework of constellation design is formalized by defining 
the notion of a lattice code, and several efficient methods for 
constructing lattice codes are presented in the literature in the 
context of multidimensional coded modulation and  MIMO systems 
\cite{constellation4}-\cite{lc3}. In the present work, we consider 
the problem of constellation design for MBM, which is a recently 
proposed {\color{black} MIMO} channel modulation scheme, described 
briefly in the earlier 
part of this section. The constellation design for MBM is fundamentally 
different from those described above mainly because of the sparse 
nature of the MBM signal vectors. Although some techniques  can be 
borrowed from the previous literature, only marginal gains are possible 
if sparsity is not explicitly taken into account while designing the 
constellation. Specifically, we need good {\em structured sparse lattice 
codes} in higher dimensions to achieve high SNR efficiency in MBM. 
To this end, the contributions of this paper {\color{black} can be 
summarized} as follows. 
\begin{itemize}
\item 
{\color{black}
The problem of constellation design with constraints on the sparsity of 
signal vectors has not been addressed in the conventional MIMO literature. 
Since the sparsity arises naturally in MBM, we formulate the constellation 
design problem for MBM by explicitly considering the sparsity constraints 
to achieve improved distance properties.} 
Specifically, we consider the design of {\color{black} structured} sparse 
multidimensional constellation, where the elements of the constellation 
are the joint MBM vectors to be transmitted in $N$ channel uses, i.e., 
MBM blocks formed by concatenating $N$ MBM vectors. 

\item Next, we show that using non-binary block codes 
\cite{nonbinary} in conjunction with the lattice based constructions 
\cite{coset1},\cite{coset2}, it is possible to design constellation 
(signal set\footnote{The terms constellation and signal set are 
interchangeably used in the paper.}) for MBM with superior distance 
properties. We give one such construction of the MBM signal set using 
the notions of  {\em MAP-index coding} and {\em multilevel squaring 
construction}.  
\item We derive upper bound on the bit error rate (BER) and the 
asymptotic diversity gain achieved by MBM using the proposed signal set, 
both of which are verified by simulations. 
\item Finally, we present simulation results that demonstrate the SNR 
gain achieved by the proposed signal set. For example, an MBM system of 
rate 3 bpcu using the proposed signal set achieves a BER of $10^{-4}$ at 
an SNR of 6 dB, while the conventional MBM signal set requires 12 dB to 
achieve the same BER performance. The improved distance properties of the 
proposed signal set are also numerically demonstrated.   
\end{itemize}
 
The rest of the paper is organized as follows. Section \ref{sec2} 
provides the necessary preliminaries for the squaring construction, 
which is used in the latter section for MBM signal set design. The 
formulation of the signal set design problem, the  proposed signal 
set, and its distance properties are presented in {\color{black} Sec.} 
\ref{sec3}. The BER upper bound and the diversity analysis of MBM 
using the proposed signal set are presented in Sec. \ref{sec4}. 
Results and discussions are presented in \ref{sec5}, and conclusions 
are presented in Sec. \ref{sec6}.

\section{Preliminaries}   
\label{sec2}
In this section, we present the notions of partitions, partition chains, 
partition distance lemma, and the squaring construction 
\cite{coset1},\cite{coset2}. These notions will be used in the next 
section for the MBM constellation design. 

Consider a discrete, finite set $S$. An $M$-way {\em partition} of the 
set $S$ is specified by $M$ disjoint sets $T(b)$, such that their union 
is the set $S$. Here, $b$ is the {\em label} for the subset $T(b)$, 
which uniquely identifies the subset. This can be an integer labeling, 
in which case $b$ can take values  $0,1,\cdots,M-1$. If $M=2^K$, for 
some $K \in \mathbb{Z}_+$, then we can use binary labelings of $K$ 
bits. The partition of $S$ into $T(b)$ is denoted by $S/ T$, with the 
{\em order} of partition  $|S/T|=M$. For example, consider the set 
$S=\{-4,-3,\cdots,0,\cdots, 3,4,5\}$, which is a subset of integers. 
A two way partition of this set into odd and even integers is 
$T(0)=\{-3,-1,1,3,5\}$ and $T(1)=\{-4,-2,0,2,4\}$. The order of 
partition here is $|S/T|=2.$ 

An $m$-level partition is denoted by $S_0/S_1/ \cdots /S_m$ and is 
obtained by first partitioning $S_0$ into $S_1(b_0),b_0=0,\cdots,M_1-1$, 
then partitioning each $S_1(b_0)$ into $S_2(b_1),b_1=1,\cdots,M_2-1$, 
and so on. Here, $M_1$ is the order of the partition $S_0/S_1$, $M_2$ 
is the order of the partition $S_1/S_2$, and so on. An $m$-level partition 
is generally labeled using $m$-part label $\mathbf{b}=(b_0,b_1,\cdots,b_m)$, 
where $b_j$ is the label for the partition $S_j/S_{j+1}$. In other words, 
specifying the $m$-level partition $\mathbf{b}$ uniquely identifies a set 
from the partition $S_0/S_1/ \cdots /S_m$. Further, the subsets at $j$th 
level are identified by the first $j$ parts of the label $(b_0,\cdots,b_j)$. 
If $|S_0/S_1|=M_1$, $|S_1/S_2|=M_2, \cdots,$ $|S_{m-1}/S_{m}|=M_m$,
then the order of partition $S_0/S_m$ is the product of the orders at 
each level, i.e., $|S_0 / S_m|=M_1 M_2 \cdots M_m$. For example, 
if $|S_0/S_1|=2$ and $|S_1/S_2|=3$, then $|S_0/S_2|=2\cdot3 =6$.

Another important notion is that of a {\em distance} metric defined on a 
discrete set. If $s$ and $s'$ are two elements of $S$, then the distance 
between them is denoted by $d(s,s')$, which is equal to zero only if  
$s=s'$ and greater than zero otherwise. The {\em minimum distance} of the 
set, denoted by $d(S)$, is the minimum $d(s,s')$ for $s \neq s'$. For a 
partition $S/T$, the distance metric of $S$ carries over to its subsets 
$T(b)$. The minimum distance of the partitioned sets $T(b)$ is defined as 
the least minimum distance among $d(T(b))$. In the present work, we are 
interested in the partition $S/T$ such that $d(T) > d(S)$. If $T(b)$ and 
$T(b')$ are subsets of $S$, then the {\em minimum subset distance} 
$d(b,b')$ is equal to $d(T(b))$ if $b'=b$, otherwise, $d(b,b')$ is the 
minimum distance between the distinct elements of the subsets $T(b)$ and 
$T(b')$. We now state an important result known as the {\em partition 
distance lemma}, which gives the lower bound on the distance between any 
two subsets in an $m$-level partition chain. 

\begin{lemma}
If $S_0/S_1/\cdots /S_m$ is an $m$-level partition chain with distances 
$d(S_0)/d(S_1)/\cdots/d(S_m)$, and $S_m(\mathbf{b})$ and $S_m(\mathbf{b}')$
are subsets with multipart labels $\mathbf{b}$ and $\mathbf{b}'$, 
respectively, then the subset distance $d(\mathbf{b}, \mathbf{b}')$ is 
lower bounded by $d(S_j)$, where if $\mathbf{b} \neq \mathbf{b}'$, $j$ 
is the smallest index such that $b_j \neq b_j'$, while if $b=b'$, $j$ 
is equal to $m$.  
\end{lemma}
 
This result can be visualized in the form of a binary tree, where the 
multipart labels are associated with different branches of the tree and 
the distance between any two subsets depends on which stage the two 
subsets diverge in the binary tree.

In general, a distance measure is useful if it has {\em additivity 
property}. Specifically, if $S$ is a set of $N$-dimensional vectors 
$\mathbf{s} \in S$, then the additive property of the distance requires 
that the distance between any two vectors $\mathbf{s}$ and $\mathbf{s}'$ 
in $S$ is equal to the sum of distances between each of their components 
$s_i$ and $s'_i$, $i=1,\cdots,N$. It is well known that 
the squared Euclidean distance naturally has the additive property. 

We now present the idea of {\em squaring construction}, which gives a 
method of constructing new sets from a given set such that the constructed 
sets ensure a certain minimum distance greater than that of the set we 
start with. The squaring construction can be described as follows. 

If $S$ is a disjoint union of $M$ subsets $T(b)$, $b=0,\cdots,M-1$, then 
the squaring construction is defined as the union of the Cartesian product 
sets $T(b)\times T(b)=T^2(b)$, $b=0,\cdots,M-1$, i.e., 
$U=\bigcup_{b=0}^{M-1} T^2(b)$, which is denoted by $|S/T|^2$.
For example, if $S=\{ 0,1,2,3\}$, $T(0) = \{0,2\}$, and $T(1) = \{1,3\}$, 
then the squaring construction is the union of sets
\begin{align*}
U_1 &= T(0) \times T(0)  = \left\{ (0,0), (0,2), (2,0), (2,2)\right\} \\
U_2 &= T(1)\times T(1) =  \left\{ (1,1), (1,3), (3,1), (3,3)\right\},
\end{align*}
and the union is 
\begin{align*}
U &= U_1 \cup U_2  \\
&= \left\{ (0,0), (0,2), (2,0), (2,2), (1,1), (1,3), (3,1), (3,3)\right\}.
\end{align*}
Note that the set $U$ is a subset of the Cartesian product of $S$ with 
itself, i.e., $U \subset S\times S$. The following lemma gives an 
important property of such a construction. Specifically, it says that 
the squaring construction ensures an increased minimum distance 
\cite{coset2}.

\begin{lemma}
If $S/T$ is a partition with minimum distances $d(S)/d(T)$, then 
$U = |S/T|^2$ has a minimum distance of 
\begin{equation}
d(U) = \min \left[ d(T), 2d(S) \right].
\end{equation}
\end{lemma}
\begin{proof} 
{\em Case 1:} If two distinct elements of $U$ belong to the same set 
$T^2(b)$, then they differ in at least one coordinate, and hence have 
a distance of $d(T)$. \\
{\em Case 2:} If two distinct elements of $U$ belong to 
different $T^2(b)$s, then the two elements differ in 
both the coordinates, and hence have a minimum distance 
of $d(S)$ in each coordinate and $2d(S)$ in total. 
\end{proof}
Consider the previous example with squared Euclidean distance as the 
distance measure on $S$. With this distance measure, the minimum 
distance of $S$, $d(S)=1$. The minimum distances $d(T(0))$ and $d(T(1))$ 
are both equal to 4. Further, the minimum distance of the set $U$ 
obtained by squaring construction, $d(U) = 2 = 2d(S)$. Thus, using 
squaring construction, the minimum distance is increased from one to two. 
In the process, the dimension of the elements of the set $S$, which is 
one, is also increased in $U$ to two. 

The squaring construction can be continued iteratively on the 
resulting sets to construct new sets  with their elements in higher 
dimensions having higher minimum distance. Such a construction is 
called as the {\em multilevel squaring construction} or {\em iterated 
squaring construction}. It is interesting to note that the idea of 
squaring construction can be used to construct many good codes and 
lattices, specifically, Reed-Muller codes and the Barnes-Wall lattices, 
owing to its elegant way of increasing the distances iteratively.  
We use this idea of multilevel squaring construction in the next 
section in conjunction with non-binary block codes to construct an 
MBM signal set/constellation with very good distance properties.

\section{MBM signal set design}
\label{sec3}
In this section, we briefly review MBM system and conventional 
MBM signal set. We formulate the MBM signal set design problem by 
imposing certain conditions, which when satisfied can lead to 
improved distance properties. We propose the technique of MAP-index 
coding using non-binary block codes in conjunction with multilevel 
squaring construction (discussed in the previous section) to meet 
the imposed conditions.

\vspace{-4mm}
\subsection{Conventional MBM signal set}
\label{sec3a}
Consider an MBM system with a single transmit antenna and $m_{rf}$ RF 
mirrors placed near the transmit antenna. Then, $N_m = 2^{m_{rf}}$ MAPs 
are possible.  Each of these MAPs create different end-to-end channel 
between the transmitter and the receiver. Let the $N_m$ different MAPs 
be assigned indices from the set $\mathcal{M}=\{0,1,2,\cdots,N_m-1\}$.  
An example mapping between the elements in $\mathcal{M}$ and the MAPs 
for $m_{rf}=2$ (i.e., $N_m=4$) is shown in Table \ref{tabel1}. 
\begin{table}[h]
\centering
\begin{tabular}{|c|c|c|}
\hline 
 Mirror 1 status & Mirror 2 status & MAP index \\ 
\hline 
ON& ON &  0 \\ 
\hline 
ON & OFF &  1 \\ 
\hline 
OFF & ON &  2 \\ 
\hline 
 OFF &  OFF &  3 \\ 
\hline 
\end{tabular}   
\caption{Mapping of mirror activation patterns to indices.}
\label{tabel1}
\vspace{-3mm}
\end{table}

In a given channel use, 
one of the MAPs is selected based on $m_{rf}$ information bits and 
a symbol from a conventional modulation alphabet $\mathbb{A}$ is 
transmitted using the selected MAP. Let $\mathbb{A}_0 \triangleq 
\mathbb{A} \cup \{0\}$. Then, the conventional MBM signal set, 
$\mathbb{S}_{\text{\tiny  MBM}}$, is the set of $N_m\times 1$-sized 
signal vectors given by
\begin{align}
\vspace{-2mm}
\mathbb{S}_{\text{\tiny MBM}} =&\big{\{}\mathbf{s}_{k} \in {\mathbb A}_0^{N_m}, \  \forall k \in \mathcal{M},  \nonumber \\
& \mbox{s.t.} \ \mathbf{s}_{k}= [0 \cdots 0 \underbrace{x}_{\mbox{{\scriptsize $(k+1)$th index}}}0 \cdots 0]^T, \  x \in \mathbb{A} \big{\}}.
\label{eq:ss-mbm}
\end{align}
\vspace{-2mm}
Consider two MBM signal vectors 
\begin{equation}
\mathbf{x}_1 = 
\begin{bmatrix}
0 \\ \vdots \\ s_1 \\ 0 \\ \vdots\\  0 
\end{bmatrix}, \ \ 
\mathbf{x}_2 = 
\begin{bmatrix}
0 \\s_2 \\ \vdots \\ 0 \\ \vdots \\0 
\end{bmatrix}.
\end{equation}
The squared Euclidean distance between these two MBM signal vectors is 
$|s_1|^2+|s_2|^2$ if the positions of the non-zeros $s_1$ and $s_2$
are different, i.e., if $\mathbf{x}_1$ and $\mathbf{x}_2$ have 
different MAP indices. On the other hand, if the positions of the  
non-zeros $s_1$ and $s_2$ are the same, i.e., if $\mathbf{x}_1$ and 
$\mathbf{x}_2$ have the same MAP-index, then the distance is 
$|s_1-s_2|^2$. The minimum (squared Euclidean) distance of the 
conventional MBM signal set, $d(\mathbb{S}_{\tiny \mbox{MBM}})$ is 
then given by
\begin{equation}
d(\mathbb{S}_{\tiny \mbox{MBM}}) = \min_{s_1,s_2 \in \mathbb{A}} \{|s_1|^2+|s_2|^2,|s_1-s_2|^2\}.
\label{eq:mbm_min_dist}
\end{equation}
For example, with the BPSK modulation, the minimum distance 
$d(\mathbb{S}_{\tiny \mbox{MBM}}) = 2$, irrespective of the number of 
RF mirrors used.

\vspace{-3mm}
\subsection{Efficient signal set design for MBM}
\label{sec3b}
As just illustrated, the minimum distance between the MBM signal vectors 
is decided by the modulation alphabet, irrespective of the number of RF 
mirrors used. Therefore, as such, the distance properties of MBM can not 
be improved much except for possible marginal improvements achievable by 
the constellation shaping to make the alphabet $\mathbb{A}$ near-circular 
\cite{constellation3}. 
Therefore, we now take a different approach where 
we form a new constellation with its points being the joint MBM vectors 
to be transmitted in $N$ channel uses. That is, we consider block 
transmission of MBM, where a block of $N$ MBM vectors is considered as 
the constellation point to be  transmitted in $N$ channel uses. This is 
the approach taken in \cite{constellation4}-\cite{constellation6} to 
construct good constellations in the case of conventional modulation. 
As we show in the sequel, this allows us to design improved signal sets 
for MBM with excellent distance properties. Consider two MBM blocks 
$\mathbf{x}$ and $\mathbf{x}'$, with each block formed by concatenating  
$N$ MBM signal vectors, as shown below:

\vspace{-3mm}
{\scriptsize
\begin{equation*}
\mathbf{x} =
\begin{bmatrix}
\begin{bmatrix}
0 \\ \vdots \\ s_1 \\ 0 \\ \vdots\\  0 
\end{bmatrix} 
\\
\begin{bmatrix}
0 \\ \vdots \\ s_2 \\ 0 \\ \vdots\\  0 
\end{bmatrix} 
\\
\vdots \\
\begin{bmatrix}
0 \\ \vdots \\ s_N \\ 0 \\ \vdots\\  0 
\end{bmatrix} 
\end{bmatrix}
, \ \
\mathbf{x'} =
\begin{bmatrix}
\begin{bmatrix}
0 \\  s_1' \\ \vdots  \\ 0 \\ \vdots\\  0 
\end{bmatrix} 
\\
\begin{bmatrix}
0 \\  s_2' \\ \vdots \\ 0 \\ \vdots\\  0 
\end{bmatrix} 
\\
\vdots \\
\begin{bmatrix}
0 \\ s_N' \\ \vdots  \\ 0 \\ \vdots\\  0 
\end{bmatrix} 
\end{bmatrix}.
\end{equation*}
}

It is an obvious but important fact that any two sparse vectors are 
different if they either differ in the position of the non-zeros or 
in the value of the non-zeros even in one coordinate. If $N$ MBM 
vectors are appended to form a transmission block as in $\mathbf{x}$
above, and if we consider the collection of all such blocks as the 
signal set, then the minimum distance is governed by those blocks 
which either differ in only one position or the blocks having 
the same support (non-zero positions) but differing in only one 
non-zero value, resulting in the same minimum distance as in 
\eqref{eq:mbm_min_dist}. Therefore, it should be noted that, 
just the block transmission does not result in improved 
distance properties. However, imposing certain conditions on 
the MAP indices and the non-zeros can lead to better distance 
properties as we show next.

Consider the two MBM blocks $\mathbf{x}$ and $\mathbf{x}'$ as shown 
above. It is easy to see that the following constraints ensure higher 
distance between $\mathbf{x}$ and $\mathbf{x}'$:
\begin{enumerate}
\item The MBM blocks $\mathbf{x}$ and $\mathbf{x}'$ have higher 
distance between them when their supports differ in more number of 
positions. That is, if $\mathbf{l}=(l_1,l_2,\cdots,l_N)$ and 
$\mathbf{l}'=(l_1',l_2',\cdots,l_N')$ are the MAP indices of the $N$ 
MBM vectors of $\mathbf{x}$ and $\mathbf{x}'$, respectively, then 
the distance between $\mathbf{x}$ and $\mathbf{x}'$ is increased by 
increasing the Hamming distance between $\mathbf{l}$ and $\mathbf{l}'$. 
\item If the MBM blocks $\mathbf{x}$ and $\mathbf{x}'$ have the same
support (i.e., same MAP indices), then the distance between $
\mathbf{x}$ and $\mathbf{x}'$ can be increased by increasing 
the distance between the non-zeros of $\mathbf{x}$ and $\mathbf{x}'$. 
That is, if $\mathbf{s} = [s_1 s_2 \cdots s_N]^T$ and $\mathbf{s}' = 
[s_1' s_2' \cdots s_N']^T$ denote the vectors containing the 
non-zeros of $\mathbf{x}$ and $\mathbf{x}'$, then the distance 
between $\mathbf{x}$ and $\mathbf{x}'$ can be increased by increasing 
the distance between $\mathbf{s}$ and $\mathbf{s}'$, in the case when 
$\mathbf{x}$ and $\mathbf{x}'$ have the same support.
\end{enumerate}
Note that the first condition is the result of the sparse nature of 
the MBM signal vectors, while the second condition is the one that 
is conventionally considered in the constellation design of the 
multidimensional coded modulation and space-time MIMO systems. 
In the next subsections, we show that {\em MAP-index coding} can be used to 
achieve the first condition and the {\em multilevel squaring construction}  
can be used to achieve the second condition.

\subsection{MAP-index coding}
\label{sec3c}
As noted earlier, the MAP indices are the unique indices assigned to 
different MAPs created by the different ON/OFF combinations of RF mirrors. 
For an MBM system with $m_{rf}$ RF mirrors, there are $N_m=2^{m_{rf}}$ 
different MAPs and hence $N_m$ MAP indices, which we denoted by the set 
$\mathcal{M}=\{0,1,\cdots,N_m-1\}$. The MAP index decides the position of 
the non-zero entry in each MBM vector. Therefore, the set of MAP indices 
$(l_1,l_2,\cdots,l_N)$ corresponding to the $N$ MBM vectors of the MBM 
transmission block decides the positions of $N$ non-zeros in the MBM 
block. If $\mathbf{l}=(l_1,l_2,\cdots,l_N)$ and 
$\mathbf{l}'=(l_1',l_2',\cdots,l_N')$ are the MAP index vectors of the 
MBM blocks $\mathbf{x}$ and $\mathbf{x}'$, respectively, then, as 
mentioned in the condition $1$ above, the distance between $\mathbf{x}$ 
and $\mathbf{x}'$ can be increased by increasing the Hamming distance 
between $\mathbf{l}$ and $\mathbf{l}'$. Therefore, if block codes with 
good Hamming distance properties can be suitably adopted for selecting 
the $N$ MAP indices in an MBM block, it is possible to achieve good 
distance properties between the MBM blocks. To this end, we present 
the notion of MAP-index coding \cite{micmbm}, which can be explained 
as follows.

First, the elements of the MAP index set $\mathcal{M}=\{0,1,\cdots,N_m-1\}$ 
are used as labels for the $N_m$ elements of Galois field $\text{GF}(N_m)$. 
This establishes an one-to-one mapping between the MAP indices and the 
elements of $\text{GF}(N_m)$. An example mapping is shown in Table 
\ref{mapping} for the case when $m_{rf}=3$, and hence $N_m=8$.

\begin{table}[h]
\centering
\begin{tabular}{|c|c|c|c|c|}
\hline 
 M1 status & M2 status & M3 status & MAP index & GF$(N_m)$ \\ 
\hline 
ON& ON &  ON & 0 & 0 \\ 
\hline 
ON & ON & OFF & 1 & 1 \\ 
\hline 
ON & OFF &  ON & 2 & $X$\\ 
\hline 
 ON &  OFF & OFF & 3 & $X+1$\\ 
\hline 
OFF & ON & ON & 4 & $X^2$ \\	
\hline
OFF & ON & OFF & 5 & $X^2 + 1$\\
\hline
OFF & OFF & ON & 6 & $X^2+X$\\
\hline
OFF & OFF & OFF & 7 & $X^2 + X +1$\\
\hline
\end{tabular}   
\caption{Labeling of MAP indices to elements of GF$(2^{m_{rf}})$.}
\label{mapping}
\end{table}

Then, consider an $(N,K)$ non-binary block code on $\text{GF}(N_m)$ with 
certain Hamming distance properties. The set of all $N$-length codewords 
of this $(N,K)$ block code forms a codebook on $\text{GF}(N_m)$ with 
$N_m^K = 2^{Km_{rf}}$ codewords. Since there is one-to-one mapping 
between the elements of $\text{GF}(N_m)$ and $\mathcal{M}$, the codebook 
of the considered $(N,K)$ block code on $\text{GF}(N_m)$ induces an 
equivalent codebook on $\mathcal{M}$, with the same Hamming distance 
properties as that of the original block code on $\text{GF}(N_m)$. 
Let $\mathcal{S}_{\mbox{{\scriptsize map}}}$ denote such an $(N,K)$ 
codebook on $\mathcal{M}$. Also, let $d_H$ be the minimum Hamming 
distance of $\mathcal{S}_{\mbox{{\scriptsize map}}}$. Then, any two 
codewords $\mathbf{l}=(l_1,l_2,\cdots,l_N)$ and 
$\mathbf{l}'=(l_1',l_2',\cdots,l_N')$ will differ in at least 
$d_H$ positions. Therefore, if the codebook 
$\mathcal{S}_{\mbox{{\scriptsize map}}}$ is used as the alphabet 
for $N$ MAP indices of the MBM block, then any two blocks $\mathbf{x}$ 
and $\mathbf{x}'$ having different supports will differ in at least 
$d_H$ positions, resulting in a distance of 
\begin{equation}
d(\mathbf{x}, \mathbf{x}') =  \sum_{i=1}^{d_H} {(|s_i|^2 + |s_i|'^2)}.
\label{eq:mic_mbm_min_dist}
\end{equation}
This  distance is clearly greater than the minimum distance of the 
conventional MBM signal set in \eqref{eq:mbm_min_dist} when $d_H >1$.
 
Although MAP-index coding is able to increase the distances between 
MBM blocks with different support, in order to increase the minimum 
distance of the signal set, the distance between the signal vectors 
having the same support but differing only in non-zero values should 
also be increased. Therefore, for MBM blocks $\mathbf{x}$ and 
$\mathbf{x}'$ with the same support, the non-zero symbol vectors 
$\mathbf{s}$ and $\mathbf{s}'$ corresponding to $\mathbf{x}$ and 
$\mathbf{x}'$ should be designed such that they have certain guaranteed 
minimum distance between them. This leads us to multilevel squaring 
construction, which we present next. 

\begin{figure*}[t]
\centering
\includegraphics[width=16cm, height=6cm]{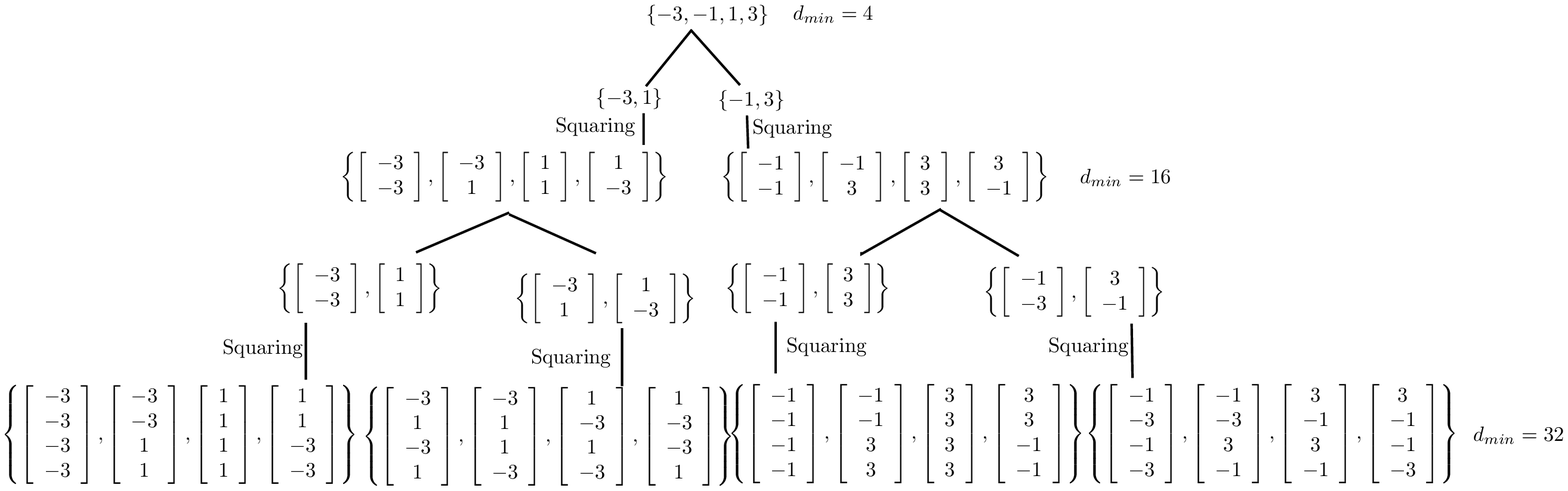}
\vspace{2mm}
\caption{Illustration of multilevel squaring construction. }
\label{fig:squaring}
\vspace{-2mm}
\end{figure*}

\subsection{Multilevel squaring construction}
\label{sec3d}

As seen in the previous subsection, when the supports of the MBM blocks 
$\mathbf{x}$ and $\mathbf{x}'$ overlap, the non-zeros of the two blocks 
$\mathbf{s}$ and $\mathbf{s}'$ should have higher distance 
($\|\mathbf{s}-\mathbf{s}'\|^2$) between them to ensure higher distance 
between $\mathbf{x}$ and $\mathbf{x}'$. As seen in Sec. \ref{sec2}, the 
squaring construction allows us to construct a set of vectors with 
certain assured minimum distance, starting from a set of scalars. 
Therefore, the squaring construction is ideally suited to construct 
signal constellations with good distance properties. In the present work, 
we use the multilevel squaring construction starting from an $M$-PAM 
alphabet and construct an $N$ dimensional signal set $\mathcal{A}$ with 
good distance properties. 

Figure \ref{fig:squaring} shows the tree representation of a two stage 
squaring construction starting from 4-PAM alphabet 
$\mathbb{A}=\{ -3,-1,1,3\}$. Each stage of the squaring construction 
consists of two steps, viz., the set partition step followed by the 
Cartesian product (which we have indicated as `squaring' in Fig. 
\ref{fig:squaring}). In each stage of the construction, we have shown 
the minimum distance between the vectors of that stage. The minimum 
distance of the original PAM signal set is $d_{min}=4$. Then, the 
squaring construction as described in Sec. \ref{sec2} is applied on 
this set to get a two dimensional signal set with a minimum distance 
of $d_{\mbox{{\scriptsize min}}}=16$. The dimension of the signal set 
is increased from one to two after the first stage of squaring 
construction. Now, the squaring construction is repeated on the 
resulting two dimensional signal set to obtain a four dimensional 
signal set with minimum distance of $d_{\mbox{{\scriptsize min}}}=32$. 
Now, if we stop the squaring construction at this stage, we have a 
constellation in four real dimensions. In order to make this signal 
set compatible with the QAM modems, we convert this four dimensional 
real vectors into two dimensional complex vectors by forming 
complex symbols using two consecutive real symbols. For example, 
the two dimensional complex vectors formed using the four dimensional 
real vectors in Fig. \ref{fig:squaring} are given by

{\footnotesize
\begin{equation*}
\mathcal{A} = \Bigg\{ 
\begin{bmatrix}
-3-3i \\ -3-3i
\end{bmatrix},
\begin{bmatrix}
-3-3i \\ 1+1i
\end{bmatrix},
\begin{bmatrix}
1+1i \\ 1+1i
\end{bmatrix},
\begin{bmatrix}
1+1i \\ -3-3i
\end{bmatrix},
\begin{bmatrix}
-3+1i \\ -3+1i
\end{bmatrix},
\end{equation*}
\begin{equation*}
\begin{bmatrix}
-3+1i \\ 1-3i
\end{bmatrix},
\begin{bmatrix}
1-3i \\ 1-3i
\end{bmatrix},
\begin{bmatrix}
1-3i \\ -3+1i
\end{bmatrix},
\begin{bmatrix}
-1-1i \\ -1-1i
\end{bmatrix},
\begin{bmatrix}
-1-1i \\ 3+3i
\end{bmatrix},
\begin{bmatrix}
3+3i \\ 3+3i
\end{bmatrix},
\end{equation*}
\begin{equation*}
\begin{bmatrix}
3+3i \\ -1-1i
\end{bmatrix},
\begin{bmatrix}
-1-3i \\ -1-3i
\end{bmatrix},
\begin{bmatrix}
-1-3i \\ -3-1i
\end{bmatrix},
\begin{bmatrix}
3-1i \\ 3-1i
\end{bmatrix},
\begin{bmatrix}
3-1i \\ -1-3i
\end{bmatrix}
\Bigg\}.
\end{equation*}
}
It should be noted that the distance properties of this set of two 
dimensional complex vectors is same as that of the set of four dimensional 
real vectors in Fig. \ref{fig:squaring}. In the present work, the 
vectors of the complex constellation $\mathcal{A}$ obtained from the 
squaring construction are used for the non-zero parts of the MBM blocks. 
This makes sure that, whenever the supports of the MBM blocks overlap, 
the distance between the MBM blocks is high.
By counting across the branches of the tree representation of the 
squaring construction, the number of vectors at the end of $L$th stage, 
starting from an $M$-PAM alphabet, is given by
\begin{equation}
|\mathcal{A}| =
\begin{cases}
 2^L \left(\frac{M^{2^L}}{2^{2(2^L-1)}} \right) \text{ if } M \geq 4 \\
  2  \ \ \ \ \ \ \ \ \ \ \ \ \ \ \ \ \text{ if } M=2.
\end{cases}
\end{equation}
Therefore, with $M=2^P$, $P \geq 2$, the number of vectors 
in $\mathcal{A}$ is given by
\begin{equation}
|\mathcal{A}| = 2^{2^L(P-2) + L+2}.
\end{equation}
It should be further noted that the $L$ level squaring construction 
results in $2^L$ dimensional real vectors and hence $N=2^L/2=2^{L-1}$ 
dimensional equivalent complex vectors.   

\subsection{Proposed signal set}
\label{sec3e}
In this subsection, we present the proposed signal set for MBM by 
putting together the ideas discussed so far.  As seen from the previous 
subsections, the MAP-index coding results in a codebook, which we 
denote by $\mathcal{C}_b$, consisting of $N$-length codewords on 
$\mathcal{M}$. This codebook can be thought of as the signal set from 
which the MAP-index vectors $\mathbf{l}=(l_1,l_2,\cdots,l_N)$ for MBM 
blocks are selected. Further, the multilevel squaring construction 
results in a multidimensional signal set $\mathcal{A}$, whose elements 
are used for the non-zero part of the MBM blocks. The proposed signal 
set for MBM is the combination of the two signal sets $\mathcal{C}_b$ and 
$\mathcal{A}$, and is given by 
\begin{equation}
\mathbb{S}=\{ \mathbf{x}=[\mathbf{x}_1^T \mathbf{x}_2^T \cdots \mathbf{x}_N^T ]^T \mbox{ s.t } (l_1,\cdots,l_N) \in \mathcal{C}_b, \ \mathbf{s} \in \mathcal{A}\},
\label{eq:prop_ss}
\end{equation}
where $\mathbf{x}_1,\cdots,\mathbf{x}_N$ are the $N$ MBM vectors which 
form the MBM block $\mathbf{x}$, $(l_1,\cdots, l_N)$ are the MAP indices 
of $\mathbf{x}_1,\cdots,\mathbf{x}_N$, respectively, and
$\mathbf{s} = [s_1,\cdots,s_N] \in \mathcal{A}$ is such that $s_i$ is the 
non-zero symbol (the source symbol) of $\mathbf{x}_i$ (the $i$th MBM 
vector of the MBM block).

{\color{black}
{\em Example:} 
For a system with $N=4$, $K=2$, $m_{rf}=3$, and $M=2$, the signal set 
$\mathcal{A}$ consisting of 2 vectors generated by squaring 
construction is

\vspace{-3mm}
{\small
\begin{equation}
\mathcal{A} = \left\{ 
\begin{bmatrix}
-1-1i \\ -1-1i \\ -1-1i\\ -1-1i
\end{bmatrix}, 
\begin{bmatrix}
1+1i \\ 1+1i \\ 1+1i\\ 1+1i
\end{bmatrix}
\right\},
\end{equation}
}

\vspace{-3mm}
\hspace{-5.5mm}
and the MAP-index codebook $\mathcal{C}_b$ generated by $(4,2)$ shortened 
Reed-Solomon code on GF(8), consisting of $2^{Km_{rf}}=2^6=64$ codewords, 
is given by 

\vspace{-1mm}
{\small
\begin{equation}
\mathcal{C}_b = \left\{ 
\begin{bmatrix}
0 \\0 \\0 \\0
\end{bmatrix}, 
\begin{bmatrix}
0 \\ 1 \\ 6 \\ 3
\end{bmatrix},
\begin{bmatrix}
0 \\ 2 \\ 7 \\ 6
\end{bmatrix},
\begin{bmatrix}
0 \\ 3 \\ 1 \\ 5
\end{bmatrix},
\begin{bmatrix}
0 \\ 4 \\ 5 \\ 7
\end{bmatrix},
\dots,
\begin{bmatrix}
7 \\ 7 \\ 3 \\ 5
\end{bmatrix}
\right\}.
\end{equation}
}

\hspace{-5.5mm}
Each vector in $\mathcal{C}_b$ is a MAP-index vector, whose entries are used 
as the MAP-indices for the  4 MBM sub-vectors which constitute the MBM block.
The MBM block is then formed by transmitting one of the vectors from 
$\mathcal{A}$ using one of the MAP-index vectors from $\mathcal{C}_b$. For 
example, if the first vector of $\mathcal{C}_b$ is used as the MAP-index 
vector and the first vector in $\mathcal{A}$ is to be transmitted, then the 
MBM transmit block corresponding to this combination is a $NN_m=32$-length 
signal vector given by 
\begin{align}
\mathbf{x} = [&(-1-i)\ 0\ 0\ 0\ 0\ 0\ 0\ 0\ (-1-i)\ 0\ 0\ 0\ 0\ 0\ 0\ 0\  \nonumber  \\ 
&(-1-i)\ 0\ 0\ 0\ 0\ 0\ 0\ 0\ (-1-i)\ 0\ 0\ 0\ 0\ 0\ 0\ 0]^T.
\label{eq:ex_sigvec1}
\end{align}

Likewise, if the second vector of $\mathcal{C}_b$ is used as the MAP-index
vector and the second vector in $\mathcal{A}$ is to be transmitted, then the
signal vector is given by
\begin{align}
\mathbf{x} = [&(1+i)\ 0\ 0\ 0\ 0\ 0\ 0\ 0\ 0\ (1+i)\ 0\ 0\ 0\ 0\ 0\ 0\ \nonumber \\
& 0\ 0\ 0\ 0\ 0\ 0\ (1+i)\ 0\ 0\ 0\ 0\ (1+i)\ 0\ 0\ 0\ 0]^T.
\label{eq:ex_sigvec2}
\end{align}
The proposed MBM constellation $\mathbb{S}$ is the set of all such MBM 
blocks obtained by different combinations of MAP-index vectors from 
$\mathcal{C}_b$ and symbol vectors from $\mathcal{A}$.
}

The number of MBM blocks in $\mathbb{S}$ (i.e., the number of signal 
points in the proposed constellation) is 
$|\mathbb{S}|=|\mathcal{C}_b||\mathcal{A}|$. As seen before, 
$|\mathcal{C}_b|=N_m^K = 2^{Km_{rf}}$ and 
$|\mathcal{A}|=2^{2^L(P-2)+L+2}$ for $P\geq 2$, and $|\mathcal{A}|=2$ 
for $P=1$.
Therefore, the rate achieved by transmitting an MBM block from the proposed 
signal set $\mathbb{S}$ for $P\geq 2$ is given by
\begin{align}
\eta &= \frac{1}{N}\left\{ \log_2\left(2^{Km_{rf}} \cdot 2^{2^L(P-2) + L+2}\right) \right\} \nonumber \\
&=\frac{1}{N}\left[Km_{rf} + 2^{L}(P-2) + L+2\right] \nonumber \\  
 &=\frac{1}{N} \bigg[ Km_{rf} + 2N(\log_2M-2) + \log_2(2N) + 2 \bigg]. 
\label{eq:rate_highM} 
\end{align}
For the case when $ P=1$, the rate is given by  
\begin{align}
\eta &=\frac{1}{N}\left\{ \log_2 \left(2^{Km_{rf}} \cdot 2\right) \right\} \nonumber \\
&=\frac{1}{N}\left[Km_{rf} + 1\right].  
\end{align}
For example, if $N=4$, $K=2$, $m_{rf}=2$, $M=4$, the achieved rate is 
\begin{align*}
\eta &= \frac{1}{4}\bigg[ 2\cdot 2 + 2\cdot 4(\log_24 - 2) + \log_28+ 2 \bigg] \nonumber \\
&=\frac{9}{4} = 2.25 \ \ \ \text{bpcu}.
\end{align*}
Further, the minimum distance of the signal set in \eqref{eq:prop_ss} is given by
\begin{equation}
d(\mathbb{S}) = \min_{\mathbf{s}, \mathbf{s}' \in \mathcal{A} } \left\{ \sum_{i=1}^{d_H} {(s_i^2 + s_i'^2)} ,
\|\mathbf{s} - \mathbf{s}'\|^2 \right\},
\end{equation}
where the first term inside the brackets is the distance between 
the blocks when the supports are non-overlapping and the second term 
corresponds to the case when the supports are overlapping. 

\subsection{The received signal} 
In this subsection, we present the expression for the received signal when 
an MBM block from the proposed signal set \eqref{eq:prop_ss} is transmitted.
We assume a frequency flat Rayleigh fading channel which remains constant 
for the duration of $N$ channel uses. An MBM block  transmitted in $N$ 
channel uses can be written in the matrix form as 
an $N_m \times N$ matrix 
$\mathbf{X} = [\mathbf{x}_1 \cdots \mathbf{x}_N]$, 
such that $\text{vec}(\mathbf{X}) = \mathbf{x} \in \mathbb{S}$. 
The received $n_r \times N$ matrix in $N$ channel uses is then given by 
\begin{equation}
\mathbf{Y} = \mathbf{HX} + \mathbf{N},
\label{eq:rxd_matrix}
\end{equation}
where $\mathbf{H} \in \mathbb{C}^{n_r \times N_m}$ is the 
MBM channel matrix, whose entries are assumed i.i.d 
$\mathcal{CN}(0,1)$ and $\mathbf{N} \in \mathbb{C}^{n_r \times N}$ 
is the additive white Gaussian noise matrix with its entries being i.i.d 
$\mathcal{CN}(0, \sigma^2)$. The received signal in 
\eqref{eq:rxd_matrix} can be written in vector form as 
\begin{equation}
\mathbf{y}=(\mathbf{I} \otimes \mathbf{H}) \mathbf{x} + \mathbf{n},
\end{equation}
where $\mathbf{y} = \text{vec}(\mathbf{Y})$, 
$\mathbf{x} = \text{vec}(\mathbf{X}) \in \mathbb{S}$, 
and $\mathbf{n} = \text{vec}(\mathbf{N})$. 
The maximum-likelihood (ML) detection rule for  
signal detection  with the proposed signal set is 
\begin{equation}
\widehat{\mathbf{x}} = \argmin_{\mathbf{x} \in \mathbb{S}} \|\mathbf{y} - (\mathbf{I} \otimes \mathbf{H}) \mathbf{x} \|^2.
\label{eq:MLD}
\end{equation}

\section{BER and asymptotic diversity analyses}
\label{sec4}
From \eqref{eq:prop_ss} and \eqref{eq:rxd_matrix}, it can be seen 
that there is certain dependence in time among the transmit MBM 
vectors of an MBM block and that the transmit block can be written 
in the form of an $N_m\times N$ space time codeword. Therefore, it 
is of interest to study if the proposed signal set can provide any 
diversity  gain apart from improved distance properties. To this end, 
in this section, we carry out the BER and asymptotic diversity analyses 
of MBM using the proposed signal set.  
 
Consider two MBM blocks $\mathbf{X}^{(i)}$ and $\mathbf{X}^{(j)}$, 
such that $\text{vec}(\mathbf{X}^{(i)})$, $\text{vec}(\mathbf{X}^{(j)}) 
\in \mathbb{S}$. The pairwise error probability (PEP) between   
$\mathbf{X}^{(i)}$ and $\mathbf{X}^{(j)}$, given the channel 
$\mathbf{H}$, is the probability that $\mathbf{X}^{(i)}$ is transmitted 
and it is detected as $\mathbf{X}^{(j)}$ at the receiver, with the 
channel being known to the receiver. This probability is given by
\begin{align}
\text{PEP}_{|\mathbf{H}} (\mathbf{X}^{(i)}, \mathbf{X}^{(j)}) &= \text{Pr}\{\mathbf{X}^{(i)} \rightarrow 
\mathbf{X}^{(j)} | \mathbf{H} \} \nonumber \\
&=\text{Pr}\{\|\mathbf{Y} - \mathbf{H}\mathbf{X}^{(i)} \|^2_F \leq 
\|\mathbf{Y} - \mathbf{H}\mathbf{X}^{(j)} \|^2_F \}, 
\label{eq:PEP1}
\end{align}
where 
$\|\cdot\|^2_F$ denotes the Frobinius norm of a matrix. The PEP in 
\eqref{eq:PEP1} can be simplified as \cite{stbc1},\cite{stbc2}
\begin{equation}
\text{PEP}_{|\mathbf{H}} (\mathbf{X}^{(i)}, \mathbf{X}^{(j)}) = Q(\sqrt{\rho D_{ij}/2}),
\label{eq:PEP2}
\end{equation}
where $Q(\cdot)$ denotes the Q-function, 
$D_{ij} \triangleq \|\mathbf{H}(\mathbf{X}^{(i)}-\mathbf{X}^{(j)})\|^2_F$,
and $\rho$ is the signal-to-noise ratio (SNR) per receive branch at the 
receiver. The PEP in \eqref{eq:PEP2} can be upper bounded as 
\begin{equation}
\text{PEP}_{|\mathbf{H}} (\mathbf{X}^{(i)}, \mathbf{X}^{(j)}) \leq \frac{1}{2}e^{-\rho D_{ij}/4},
\end{equation}
where we have used the inequality $Q(x) \leq \frac{1}{2}e^{-x^2/2}$. 
For i.i.d Gaussian channels, unconditioning the PEP over the 
channel results in the following inequality \cite{stbc1},\cite{stbc2}
{\small
\begin{align}
\text{PEP}(\mathbf{X}^{(i)}, \mathbf{X}^{(j)})\hspace{-0.5mm} &\leq  
\hspace{-0.5mm}\frac{1}{2} \hspace{-1mm}\left( \hspace{-1mm} \frac{1}{\text{det}\left(\mathbf{I}_{N_m}\hspace{-1.0mm} + \hspace{-0.5mm}
\frac{\rho}{4} (\mathbf{X}^{(i)} \hspace{-0.5mm} - \hspace{-0.5mm} \mathbf{X}^{(j)})
(\mathbf{X}^{(i)} \hspace{-0.5mm} - \hspace{-0.5mm} \mathbf{X}^{(j)})^H \right) } \hspace{-1mm}\right)^{\hspace{-1mm}n_r}\nonumber \\
&=\frac{1}{2}\left( \frac{1}{\prod_{r_{ij}=1}^{R_{ij}} (1+\sigma_{r_{ij}}^2\rho/4)} \right)^{\hspace{-1.0mm}n_r},
\label{eq:PEP3}
\end{align}}

\vspace{-1mm}
\hspace{-4.75mm}
where $\mathbf{I}_{N_m}$ is the $N_m \times N_m$ identity matrix, 
$\sigma_{r_{ij}}$ is the $r_{ij}$th singular value of $(\mathbf{X}^{(i)} - 
\mathbf{X}^{(j)})$, and $R_{ij}$ is its rank.
The union bound based BER upper bound for the proposed 
signal set is then given by
{\small
\begin{align}
\text{BER} &\leq \frac{1}{|\mathbb{S}|} \sum_{i=1}^{|\mathbb{S}|} 
\sum_{j=1, j\neq i}^{|\mathbb{S}|} \text{PEP}(\mathbf{X}^{(i)}, \mathbf{X}^{(j)}) \frac{d(\mathbf{X}^{(i)},\mathbf{X}^{(j)})}{\kappa} \nonumber \\
&=\frac{1}{|\mathbb{S}|} \sum_{i=1}^{|\mathbb{S}|} 
\sum_{j=1, j\neq i}^{|\mathbb{S}|}\hspace{-1mm} \frac{1}{2}\hspace{-1mm}\left( \frac{1}{\prod_{r_{ij}=1}^{R_{ij}} (1+\sigma_{r_{ij}}^2\rho/4)} \right)^{\hspace{-1mm}n_r} \hspace{-2mm} \frac{d(\mathbf{X}^{(i)},\mathbf{X}^{(j)})}{\kappa},
\label{eq:BER}
\end{align}}

\vspace{-2mm}
\hspace{-4.75mm}
where $d(\mathbf{X}^{(i)},\mathbf{X}^{(j)})$ is the Hamming distance 
between the bit mappings of $\mathbf{X}^{(i)}$ and $\mathbf{X}^{(j)}$, 
and $\kappa=\log_2|\mathbb{S}|$.\\

\begin{theorem}
The asymptotic diversity order of MBM using the proposed signal set is $n_r$.
\label{theorem1}
\end{theorem}
\begin{proof}
At high SNR values, the PEP in \eqref{eq:PEP3} can be simplified as 
\begin{equation}
\text{PEP}(\mathbf{X}^{(i)}, \mathbf{X}^{(j)}) \leq \left(\frac{\rho}{4}\right)^{-n_r R_{ij}} \left( \left( \prod_{r_{ij}=1}^{R_{ij}}\hspace{-1mm}\sigma_{r_{ij}}^2 \right)^{\hspace{-1mm}1/R_{ij}} \right)^{\hspace{-1mm}-n_rR_{ij}}\hspace{-1mm}.
\label{eq:PEP4}
\end{equation}
From \eqref{eq:PEP3}, the asymptotic diversity order of MBM 
using the proposed signal set is given by \cite{stbc1}
\begin{equation}
g_d = n_r \min_{i, j\neq i}R_{ij},
\label{eq:div1}
\end{equation} 
where  $R_{ij}$ is the rank of the difference matrix 
$\Delta^{ij} = \mathbf{X}^{(i)} - \mathbf{X}^{(j)}$, where $\text{vec}(\mathbf{X}^{(i)})$, 
$\text{vec}(\mathbf{X}^{(j)}) \in \mathbb{S}$. The matrices 
$\mathbf{X}^{(i)}$ and $\mathbf{X}^{(j)}$ are $N_m\times N$ matrices 
with a single non-zero entry per column. The positions of the $N$ 
non-zero entries are together determined by the MAP-index codeword, as 
discussed before. From \eqref{eq:div1}, the diversity order is determined 
by the minimum $R_{ij}$ among all $i, j\neq i$. Therefore, if we find two 
matrices $\mathbf{X}^{(i)}$ and $\mathbf{X}^{(j)}$ such that their difference 
$\Delta^{ij} = \mathbf{X}^{(i)} - \mathbf{X}^{(j)}$ has the minimum rank 
among all pairs of matrices, then that pair determines the asymptotic 
diversity order. To find such a pair, we note that any $(N,K)$ codebook 
will contain an all-zero codeword. If the zero element of GF$(2^{m_{rf}})$ 
is mapped to the MAP index `0', then the MAP-index codeword corresponding 
to the all-zero codeword is also all-zeros. This means that same MAP 
is used in all the $N$ channel uses. Consider two such transmission blocks
\begin{equation*}
X^{(i)} =
\begin{bmatrix}
s_1^i & s_2^i & \cdots & s_N^i \\
0 & 0 & \cdots & 0\\
\vdots & & \cdots & \vdots\\
 0 & 0 & \cdots & 0\\
\end{bmatrix},
X^{(j)} =
\begin{bmatrix}
s_1^j & s_2^j & \cdots & s_N^j \\
0 & 0 & \cdots & 0\\
\vdots & & \cdots & \vdots\\
 0 & 0 & \cdots & 0\\
\end{bmatrix}.
\end{equation*}
Then, their difference matrix 
\begin{equation}
\Delta^{ij} =
\begin{bmatrix}
\delta_1^{ij} & \delta_2^{ij} & \cdots & \delta_N^{ij} \\
0 & 0 & \cdots & 0\\
\vdots & & \cdots & \vdots\\
 0 & 0 & \cdots & 0\\
\end{bmatrix},
\label{eq:rank1delta}
\end{equation}
where $\delta_k^{ij} = s_k^{i} - s_k^j$, clearly has rank one. 
Therefore, $\min_{i,j\neq i} R_{ij}=1$, and hence the asymptotic 
diversity order of MBM with the proposed signal set is $n_r$. 
\end{proof}
The above result says that, although there is certain coding across time, 
the proposed signal set does not achieve transmit diversity. However, 
as we will see in the next section, MBM with the proposed signal set 
exhibits a diversity slope higher than $n_r$ in the medium SNR regime 
and the asymptotic diversity order of $n_r$ is observed only at extremely 
low BER values.

\section{Results and discussions}
\label{sec5}
In this section, we present numerical and simulation results that 
illustrate that MBM with the proposed signal set achieves improved 
distance properties resulting in good bit error performance. We 
also show that the BER upper bound derived in the previous section 
closely matches the simulated BER at high SNR values. We use this 
bound to verify the analytical diversity order derived in the 
previous section. 
  
Figure \ref{fig:ML1} shows the BER performance of MBM with the 
proposed signal set based on MAP-index coding (MIC) and squaring 
construction (SQ), which is abbreviated in the figure as MIC-SQ-MBM. 
The system considered in the figure uses $N=4$, $K=2$, $m_{rf}=4$, 
$n_r=4$, and achieves a rate of 2.25 bpcu. For MAP-index coding, the 
codebook of $(4,2)$ shortened Reed-Solomon code on $\text{GF}(2^4)$ 
is used and eight level ($2N=8$) squaring construction is achieved 
starting from $M=2$-PAM alphabet. The figure also shows the performance 
of conventional MBM signal set with rate 2 bpcu. The upper bounds on 
the BER for both the systems are also shown. From the figure 
it can be seen that the derived BER upper bound is close to the 
simulated BER at high SNR values. {\color{black} This is because the 
bound on the Q-function used for deriving the upper bound on the BER
is tight for higher values of SNR.} Further, it can be seen that 
the proposed signal set achieves superior bit error performance 
compared to conventional MBM signal set. For example, the proposed 
signal set has an SNR gain of about 7 dB at a BER of $10^{-5}$ 
compared conventional MBM signal set. 
A similar performance gain in favor of the proposed signal set is
observed in Fig. \ref{fig:ML2} for another set of parameters. In
Fig. \ref{fig:ML2}, the proposed signal set uses $N=4$, $K=2$, $m_{rf}=6$, 
$n_r=4$ and achieves a rate of 3.25 bpcu. For MAP-index coding, the 
codebook of $(4,2)$ shortened Reed-Solomon code on $\text{GF}(2^6)$ 
is used and eight level squaring construction is achieved starting 
from $M=2$-PAM alphabet. The performance of this MIC-SQ-MBM signal
set is compared with that of conventional MBM signal set with rate 
3 bpcu.

\begin{figure}[t]
\centering
\includegraphics[width=9cm, height=6cm]{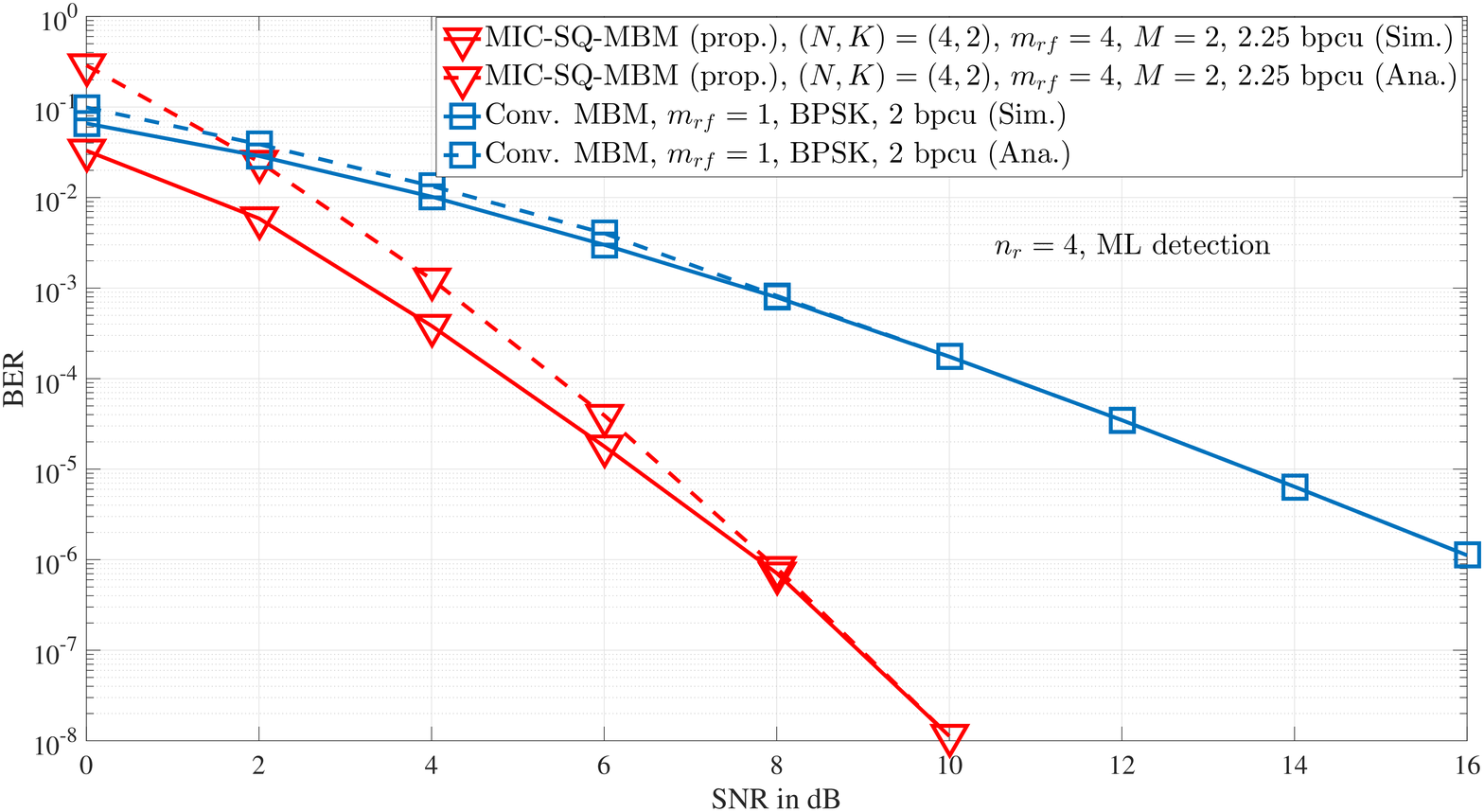}
\caption{{\color{black} BER performance of MIC-SQ-MBM (prop.) signal set 
with rate 2.25 bpcu and conventional MBM signal set with rate 2 bpcu. 
Simulation and analysis.}}
\label{fig:ML1}
\vspace{-3mm}
\end{figure}

\begin{figure}[t]
\centering
\includegraphics[width=9cm, height=6cm]{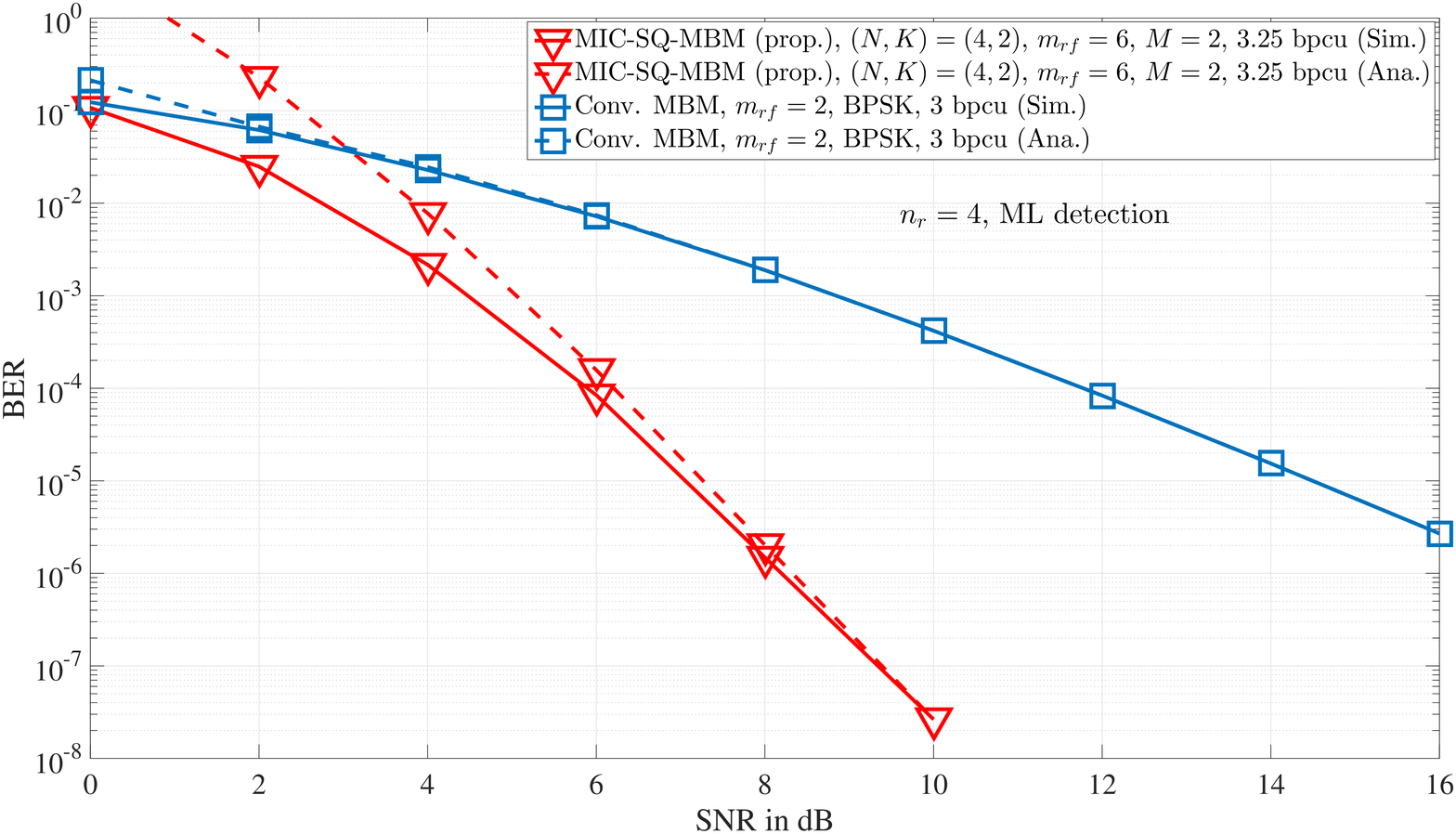}
\caption{{\color{black} BER performance of MIC-SQ-MBM (prop.) signal set 
with rate 3.25 bpcu and conventional MBM signal set with rate 3 bpcu. 
Simulation and analysis.}}
\vspace{-3mm}
\label{fig:ML2} 
\end{figure}

The superior BER performance of the MBM system with the proposed signal 
set is the result of the good distance properties achieved by the 
proposed signal set.  The distance distributions of the proposed signal 
sets and conventional MBM signal sets considered in Figs. \ref{fig:ML1} 
and \ref{fig:ML2} are shown in Tables \ref{dd2bpcu} and \ref{dd3bpcu}, 
respectively. From Table \ref{dd2bpcu}, it can be seen that, the minimum 
distance of the conventional MBM signal set is 2, while that of MIC-SQ-MBM 
(proposed) signal set is 12, which is significantly higher. Further, the 
dominant distance in MIC-SQ-MBM signal set is 16, with 76.32 \% of the 
signal pairs having this distance. A similar observation can be made 
from the Table \ref{dd3bpcu} where the distance distributions of 
conventional MBM signal set of 3 bpcu and MIC-SQ-MBM signal set of 
3.25 bpcu are shown. This explains the superior performance of the 
proposed signal set. At this point, we make few remarks below.

\begin{table}
\begin{subtable}{.5\linewidth}
\centering
\begin{tabular} {|c|c|c|}
\hline 
Distance & \# pairs & \% pairs \\
\hline 
2 & 4  & 66.67 \\ 
\hline 
4 & 2  & 33.33\\ 
\hline  
\end{tabular} 
\caption{MBM}
\end{subtable}%
\begin{subtable}{.5\linewidth}
\centering
\begin{tabular} {|c|c|c|}
\hline 
Distance & \# pairs & \% pairs\\
\hline 
12 & 15360 &  11.742\\ 
\hline 
16 & 99840 & 76.32 \\ 
\hline     
20 & 15360 & 11.742 \\ 
\hline         
32 & 256 & 0.1956 \\ 
\hline  
\end{tabular}
\caption{MIC-SQ-MBM}
\end{subtable} 
\caption{Distance distribution of MIC-SQ-MBM (prop.) signal set of rate 
2.25 bpcu and conventional MBM signal set of rate 2 bpcu, considered in 
Fig. \ref{fig:ML1}.}
\label{dd2bpcu} 
\vspace{-3mm}
\end{table}

\begin{table}
\begin{subtable}{.5\linewidth}
\centering
\begin{tabular} {|c|c|c|}
\hline 
Distance & \# pairs & \% pairs \\
\hline 
2 & 24  & 85.714 \\ 
\hline 
4 & 4  & 14.285\\ 
\hline  
\end{tabular} 
\caption{MBM}
\end{subtable}%
\begin{subtable}{.5\linewidth}
\centering
\begin{tabular} {|c|c|c|}
\hline 
Distance & \# pairs & \% pairs\\
\hline 
12 & 1032192 &  3.0765\\ 
\hline 
16 & 31481856 & 93.8346 \\ 
\hline     
20 & 1032192 & 3.0765 \\ 
\hline         
32 & 4096 & 0.0124 \\ 
\hline  
\end{tabular}
\caption{MIC-SQ-MBM}
\end{subtable} 
\caption{Distance distribution of MIC-SQ-MBM (prop.) signal set of 
rate 3.25 bpcu and conventional MBM signal set of rate 3 bpcu, considered 
in Fig. \ref{fig:ML2}.}
\label{dd3bpcu} 
\vspace{-2mm}
\end{table}

{\subsubsection*{Remark 1}
In Theorem \ref{theorem1} in the previous section, we showed that the 
asymptotic diversity order of MBM with the proposed signal set is $n_r$, 
which is same as that of MBM using conventional MBM signal set. However, 
in Figs. \ref{fig:ML1} and \ref{fig:ML2}, even at BERs as low as $10^{-6}$, 
the slopes of the BER curves for the proposed signal set and conventional 
signal set are different. Specifically, MBM with the proposed signal set 
is seen to exhibit a diversity slope higher than $n_r$. To gain more 
insight into this behavior, we plot the BER upper bounds of MBM using 
the proposed signal set for much lower BER values. The BER upper bounds 
for the previously considered systems in Figs. \ref{fig:ML1} and 
\ref{fig:ML2} are shown in  Figs. \ref{fig:ML1_bound} and 
\ref{fig:ML2_bound}, respectively. In Figs. \ref{fig:ML1_bound} and 
\ref{fig:ML2_bound}, the bounds are plotted up to a BER of $10^{-20}$. 
From these plots, it is evident that, although MBM using the proposed 
signal set initially shows a higher diversity slope, eventually the 
diversity slope becomes the same as that of MBM using conventional MBM 
signal set. For example, in Fig. \ref{fig:ML1_bound}, the curve of MBM 
with the proposed signal set changes the slope at about $10^{-12}$ BER 
and becomes parallel to the corresponding curve for MBM using conventional 
MBM signal set. {\color{black}This behavior can be explained as follows. As 
shown in the diversity analysis in Sec. \ref{sec4}, the slope of the BER 
curve in the asymptotic high SNR regime depends on the minimum rank of the 
difference matrices, $\mathbf{\Delta}^{ij}$s. The diversity analysis also 
showed that the minimum rank is always equal to one for the proposed signal 
set, which resulted in the conclusion that the asymptotic diversity order 
is $n_r$. The proposed signal set reduces the number of rank one difference 
matrices relative to the total number of possible difference matrices, which 
makes the effect of the minimum rank to show up only at much higher SNRs.
For example, for the considered system in Fig. \ref{fig:ML1_bound}, there 
are $2^9$ possible signal matrices $\mathbf{X}$ and hence there are 
$\binom{2^9}{2}=130816$ possible difference matrices. Numerical computation 
of the difference matrices and their ranks reveal that there is only one 
difference matrix which has rank one and all other difference matrices have 
higher ranks. This results in the BER curve to have a slope higher than 
$n_r$ in the low-to-medium SNRs and a change of slope to $n_r$ in the 
high SNR regime (where the diversity slope of $n_r$ manifests due to the 
effect of the presence of one difference matrix of rank one). Likewise, 
out of the $\binom{2^{13}}{2}=33550336$ possible difference matrices in 
the system considered in Fig. \ref{fig:ML2_bound}, only one difference 
matrix has rank one. 
Therefore, the slope change in Fig. \ref{fig:ML2_bound} occurs at a much 
lower probability of error of $10^{-15}$ compared to the slope change at 
a probability of error of $10^{-12}$ in Fig. \ref{fig:ML1_bound}.}
The observations in Figs. \ref{fig:ML1_bound} and \ref{fig:ML2_bound} 
convey the following points: $i)$ the numerical plots of the BER upper 
bound validate asymptotic diversity order of $n_r$ predicted by analysis 
in the previous section, and $ii)$ the proposed signal set achieves higher 
than $n_r$ diversity slope in the practical low-to-moderate SNRs of interest. 

\begin{figure}[t]
\centering
\includegraphics[width=9cm, height=6cm]{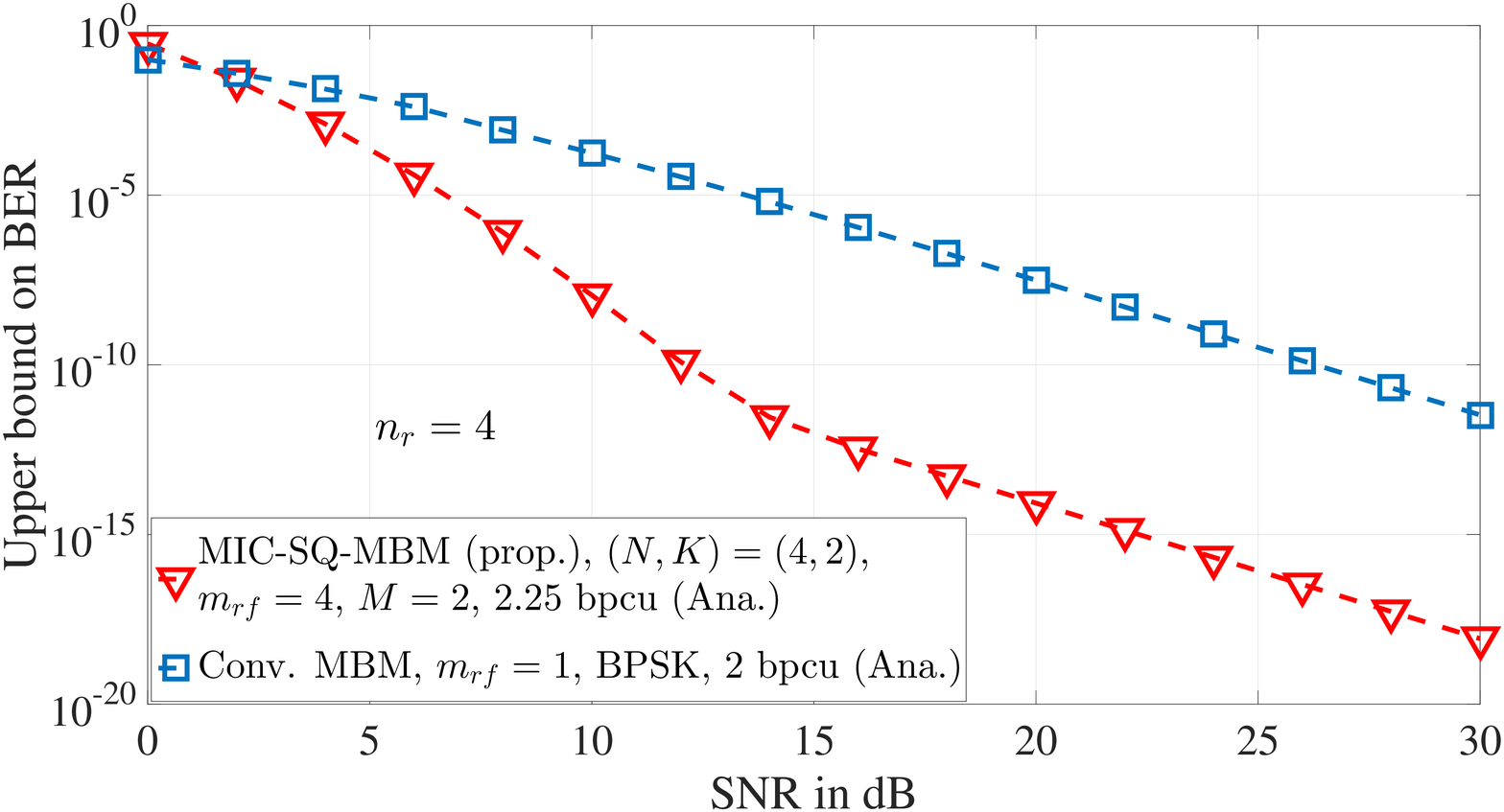}
\caption{BER upper bound plots of MBM using MIC-SQ-MBM (prop.) signal set 
and conventional MBM signal set.}
\label{fig:ML1_bound}
\vspace{-3mm}
\end{figure}

\begin{figure}[t]
\centering
\includegraphics[width=9cm, height=6cm]{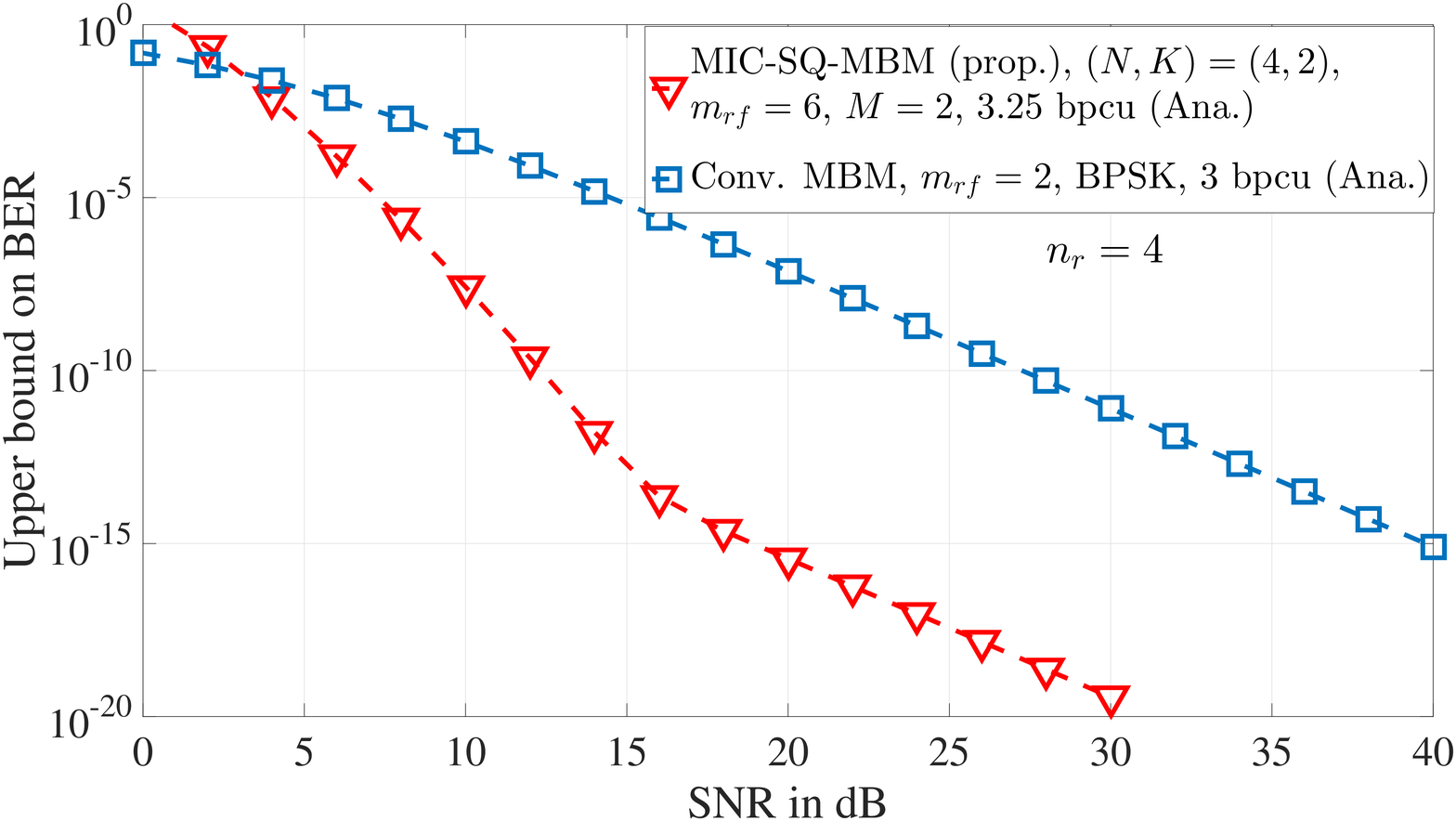}
\caption{BER upper bound plots of MBM using MIC-SQ-MBM (prop.) 
signal set and conventional MBM signal set.}
\label{fig:ML2_bound}
\vspace{-3mm}
\end{figure}

\subsubsection*{Remark 2 (Shortening of RS codes)} In the discussion of
the results of Figs. \ref{fig:ML1} and \ref{fig:ML2}, we mentioned that 
a \textit{shortened} RS code is used for MAP-index coding. In this remark, 
we give a brief account of the shortening of RS codes. An RS code on 
GF$(2^{m_{rf}})$ will have a codeword length of $N=2^{m_{rf}}-1$ 
(which is not the case in Figs. \ref{fig:ML1} and \ref{fig:ML2}). A 
shortened RS code is one in which the codeword length is less than
$2^{m_{rf}}-1$. The shortened $(N,K)$ RS code actually uses an $(N',K')$ 
RS code with $N' = 2^{m_{rf}}-1$ and $K' = K+(N'-N)$. The shortening is 
done by initially padding each message of length $K$ with $N'-N$ 
prepending zeros. RS encoding is done for this zero padded message 
to obtain a codeword with the allowed codeword length of $N'$. Finally,
the padded zeros are removed from the codeword along with puncturing 
some of the parity symbols.  

\begin{figure}[t]
\centering
\includegraphics[width=9cm, height=6cm]{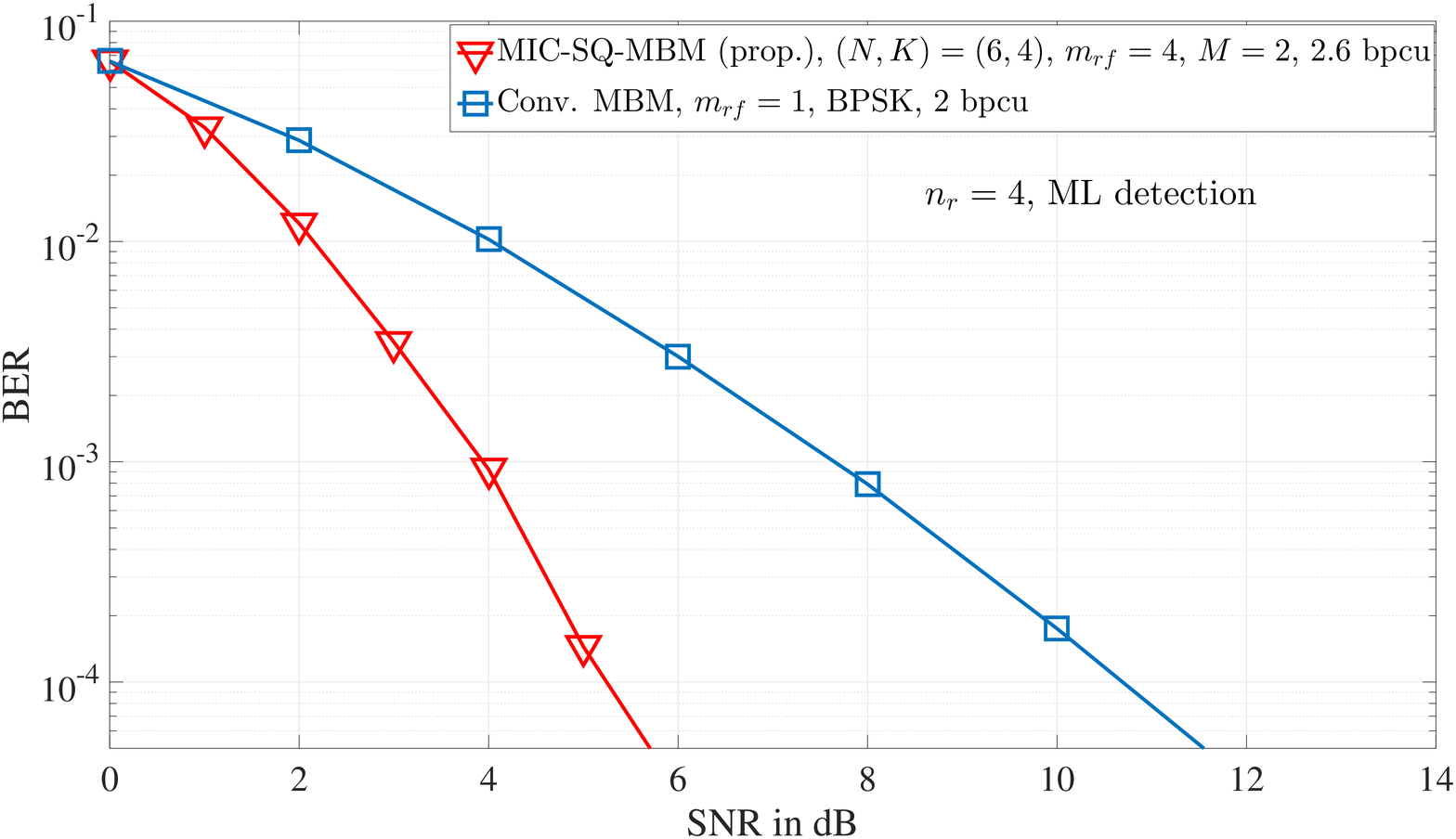}
\caption{{\color{black} BER performance of MIC-SQ-MBM (prop.) signal set 
with $N=6$, $K=4$. Performance of conventional MBM signal set is also 
shown for comparison.}}
\label{fig:ber_N6_K4}
\end{figure}

In the rest of this section, we further illustrate the advantage of the 
proposed signal set. {\color{black} Figure \ref{fig:ber_N6_K4} shows the 
BER performance of the proposed MBM signal set for $N=6$, $K=4$, 
$m_{rf}=4$, $M=2$, 2.6 bpcu, $n_r=4$, and ML detection. The performance 
of conventional MBM signal set with $m_{rf}=1$, BPSK, 2 bpcu, $n_r=4$,
and ML detection is also shown for comparison. From the figure, it can be 
seen that the proposed signal set achieves better performance compared 
to conventional MBM signal set by about 5 dB at $10^{-4}$ BER.}  
We show the BER performance of the proposed 
signal set as a function of the number of receive antennas in 
Fig. \ref{fig:ber_vs_nr}. The considered system uses $N=4$, $K=2$, 
$m_{rf}=6$, $M=2$, and achieves a rate of 3.25 bpcu. The performance 
of this system is plotted as a function of the number of receive 
antennas at two SNR values, namely, 0 dB and 2 dB. The performance 
of conventional MBM signal set is also shown for comparison. From 
Fig. \ref{fig:ber_vs_nr}, it can be seen that the proposed signal set 
requires fewer number of receive antennas compared to conventional MBM 
signal set to achieve the same bit error performance. For example, 
to achieve a BER of $10^{-4}$ at an SNR of 2 dB, the proposed signal 
set requires about 7 receive antennas, while the conventional MBM signal 
set requires 15 receive antennas. It can further be seen that this gap 
in the required number of receive antennas widens as the required BER 
goes down. 

\begin{figure}[t]
\centering
\includegraphics[width=9cm, height=6cm]{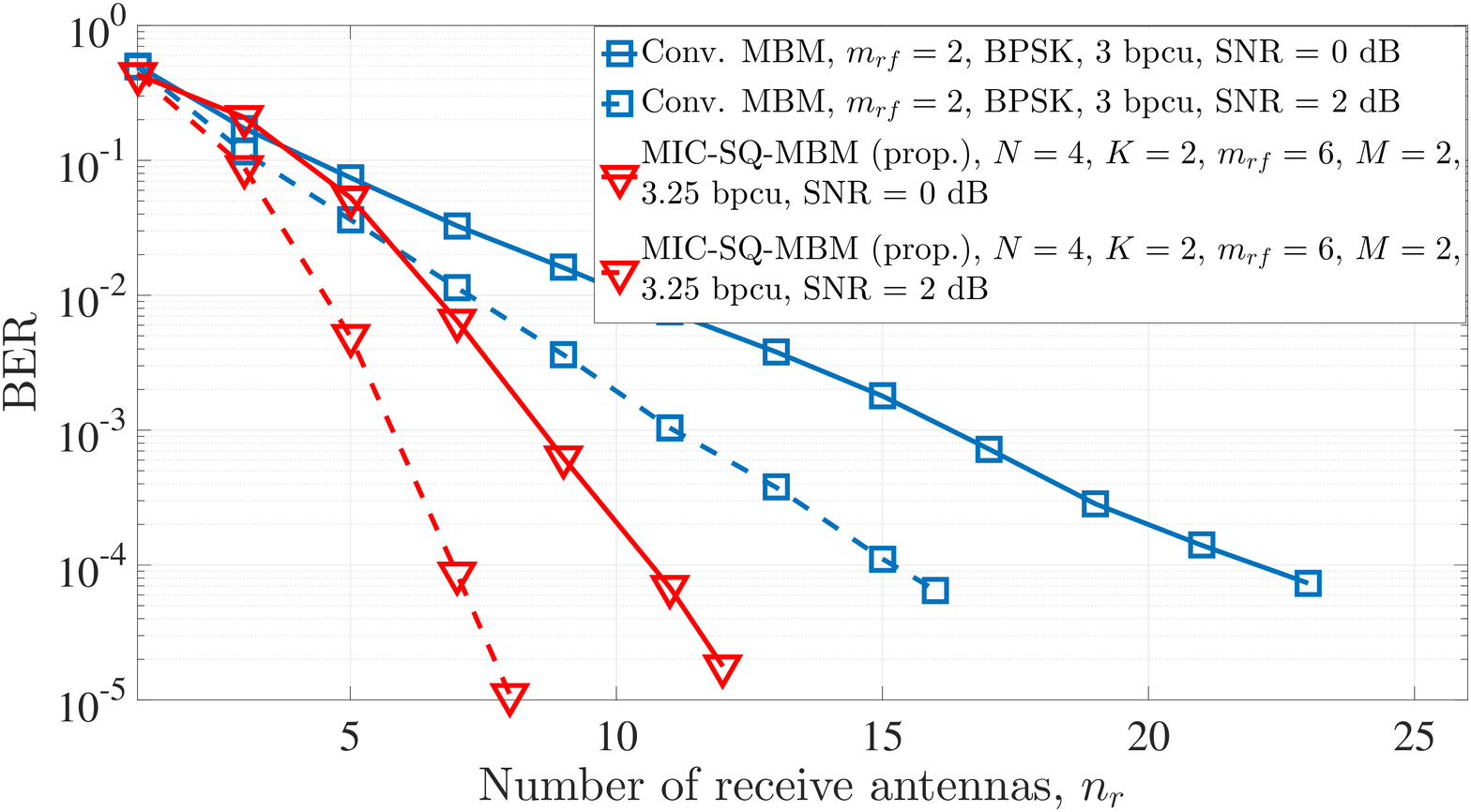}
\caption{BER performance of MBM using MIC-SQ-MBM (prop.) signal set 
as a function of number of receive antennas.}
\label{fig:ber_vs_nr}
\vspace{-3mm}
\end{figure}

In Fig. \ref{fig:snr_vs_nr}, the SNR required to achieve a BER of 
$10^{-3}$ for different number of receive antennas is shown for MBM 
with the proposed signal set (3.25 bpcu) and conventional MBM signal 
set (3 bpcu). It is evident from the figure that, for a given number 
of receive antennas, the proposed signal set achieves a BER of 
$10^{-3}$ at lesser SNR values compared to conventional MBM signal set.  
{\color{black} Figure \ref{fig:large_nr} shows the BER performance of the 
proposed MBM constellation for $n_r=4$ and $n_r=16$. The performance of 
the conventional MBM constellation is also shown for comparison. It can
be seen that the MBM performance (with proposed constellation and
conventional constellation) gets significantly better for $n_r=16$ compared
to that with $n_r=4$. This is in line with the observation in \cite{mbm2a} 
where the benefits of MBM are reported to be more pronounced for a large 
number of receive antennas and, in many cases, the conventional constellation
itself (without any channel coding) offers an acceptably low probability
of error. This can be seen in Fig. \ref{fig:large_nr}, where the conventional 
constellation achieves a BER of $10^{-5}$ at SNRs of about 13.5 dB for 
$n_r=4$ and 2.5 dB for $n_r=16$. Figure \ref{fig:large_nr} further shows 
that with the proposed constellation the same BER of $10^{-5}$ is achieved 
at SNRs of about 6.3 dB for $n_r=4$ and -2.2 dB for $n_r=16$.}

\begin{figure}[t]
\centering
\includegraphics[width=9cm, height=6cm]{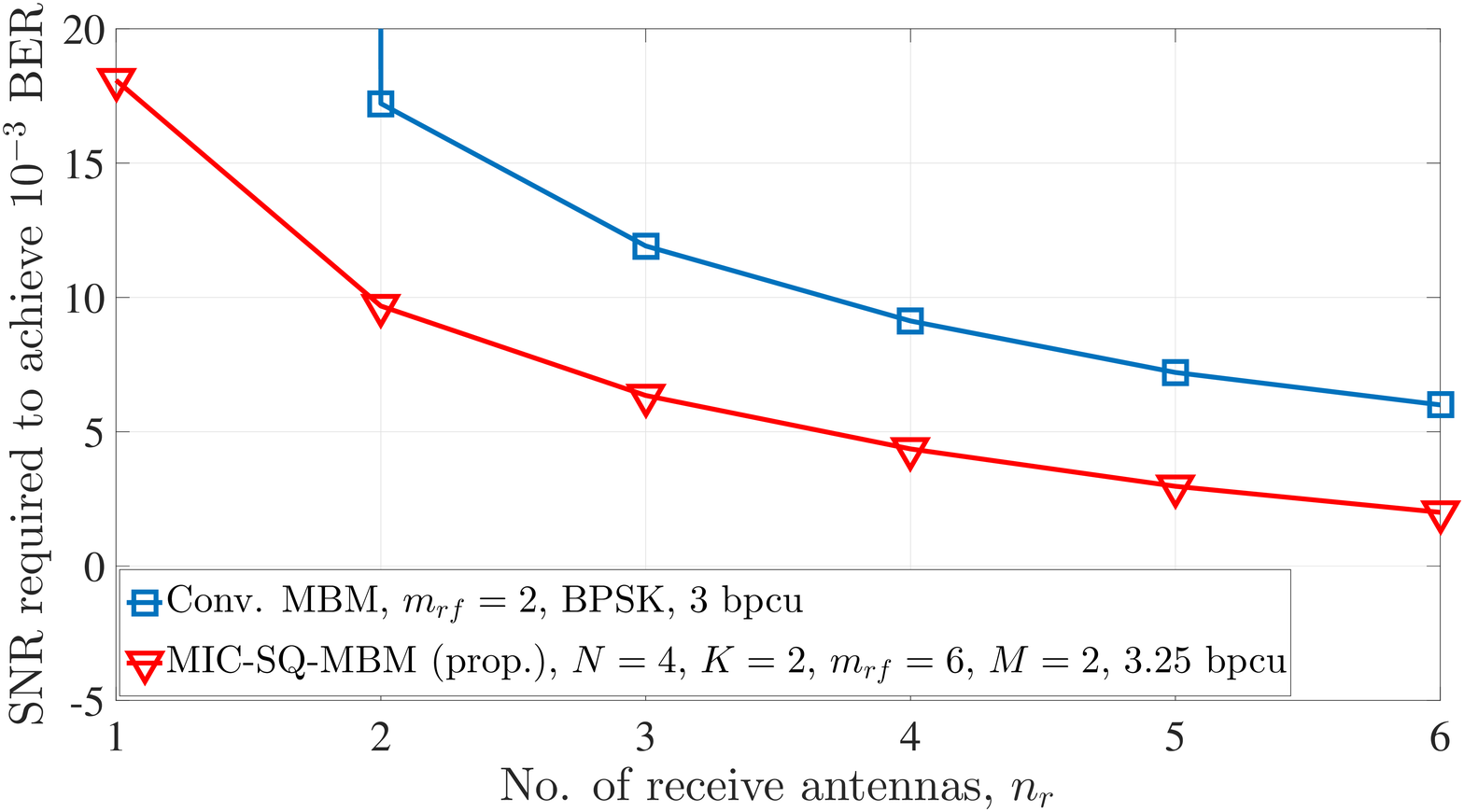}
\caption{SNR required to achieve a BER of $10^{-3}$ as a function of 
number of receive antennas for MBM using MIC-SQ-MBM (prop.) signal set 
and conventional MBM signal set.}
\label{fig:snr_vs_nr}
\end{figure}

\begin{figure}[t]
\centering
\includegraphics[width=9cm, height=6cm]{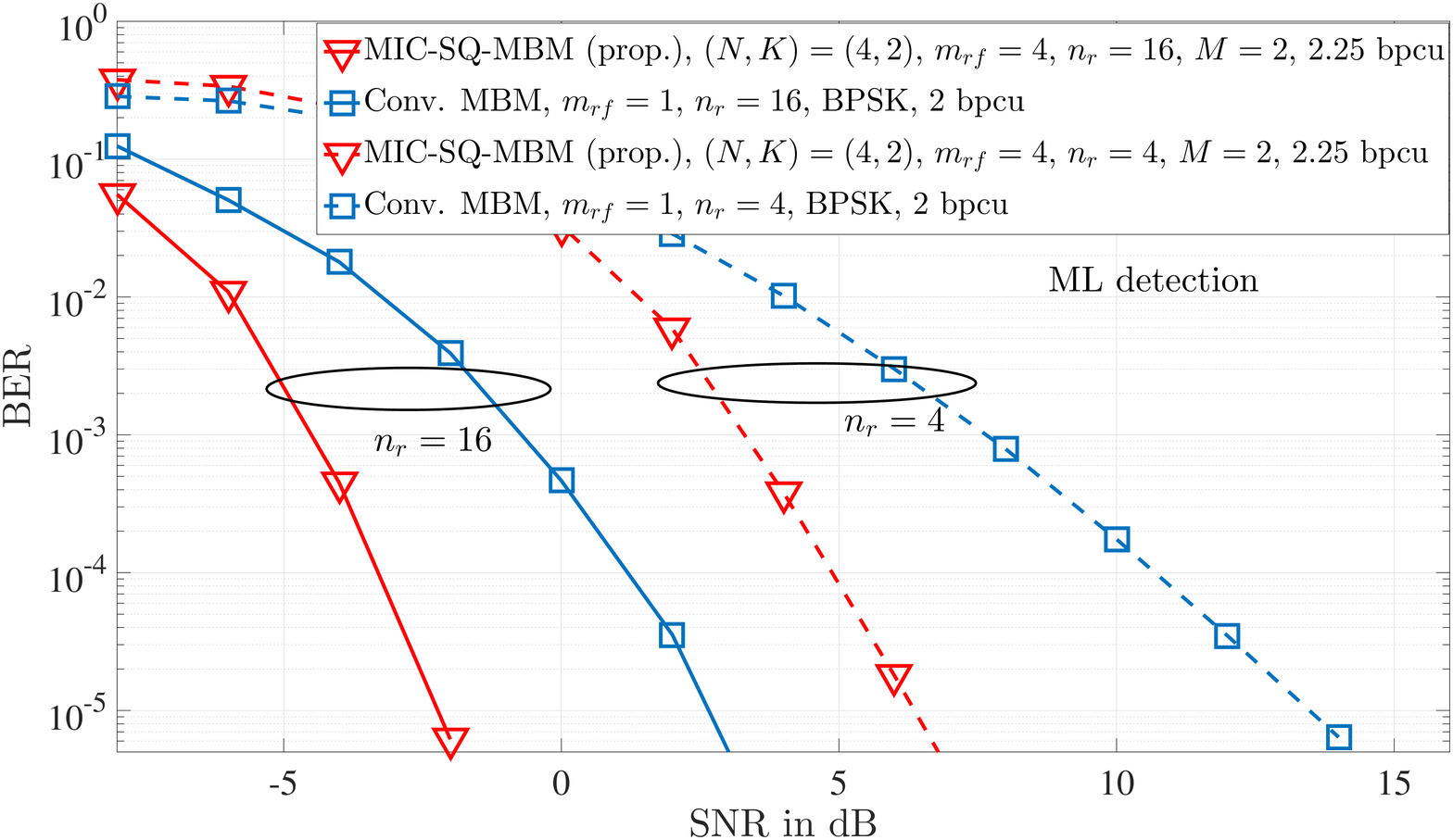}
\caption{\color{black} BER performance of MIC-SQ-MBM (prop.) signal set 
and conventional MBM signal set with $n_r=4,16$. }
\label{fig:large_nr}
\end{figure}

{\color{black} In Fig. \ref{fig:BER_vs_EbNo}, we consider two MBM systems 
using the same number of RF mirrors ($m_{rf}=4$) and show their BER 
performance as a function of $E_b/N_0$. The first system uses the proposed 
constellation with $N=4$, $K=2$, $M=2$, and achieves 2.25 bpcu. The second 
system uses conventional MBM constellation with BPSK modulation and achieves
5 bpcu. The figure shows that the system using the proposed constellation
achieves a BER of $10^{-5}$ at an $E_b/N_0$ of about 2.2 dB, whereas the
system with conventional MBM constellation achieves the same BER at an
$E_b/N_0$ of about 9.6 dB.}
\begin{figure}[t]
\centering
\includegraphics[width=9cm, height=6cm]{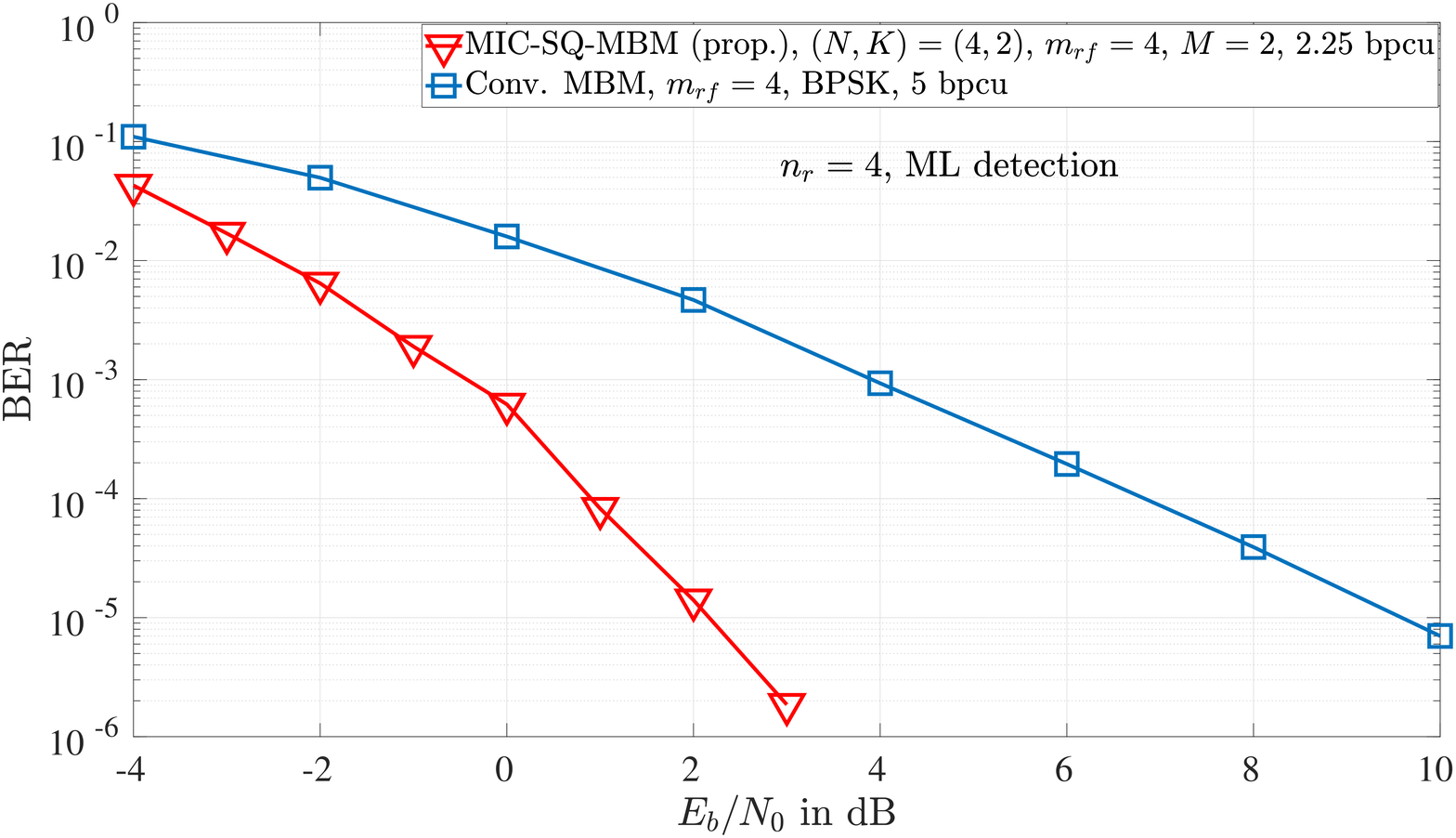}
\caption{{\color{black}BER performance comparison as a function of $Eb/N_0$ 
between the MBM systems using MIC-SQ-MBM (prop.) constellation and 
conventional MBM constellation, when both the systems use $m_{rf}=4$ 
RF mirrors.}}
\label{fig:BER_vs_EbNo}
\end{figure}}

{\color{black}
Figures \ref{fig:STBC1} and \ref{fig:STBC2} show a BER performance
comparison between MBM systems using the proposed constellation and
STBC systems. The MBM and STBC systems considered in the figures are
closely matched in terms of their achieved rates. The considered system
parameters are shown in Figs. \ref{fig:STBC1} and \ref{fig:STBC2}.
In Fig. \ref{fig:STBC1}, the MBM system has a rate of 1.75 bpcu and the
STBC system has a rate of 1.5 bpcu. 
In Fig. \ref{fig:STBC2}, the rates of both MBM and STBC systems are 2.25 
bpcu. From the figures it can be seen that the MBM systems using the 
proposed constellation achieve better BER performance compared to the 
STBC systems. For example, in Fig. \ref{fig:STBC1}, even with a slightly 
higher rate of 1.75 bpcu, the MBM system performs better by about 1 dB at 
$10^{-5}$ BER compared to the STBC system with a rate of 1.5 bpcu. In 
Fig. \ref{fig:STBC2}, for the same rate of 2.25 bpcu, the MBM system 
performs better by about 5.5 dB at $10^{-5}$ BER. It is noted that while 
the considered STBC systems require $n_t=4$ RF chains, the MBM systems 
require only a single RF chain.
}

\begin{figure}[t]
\centering
\begin{subfigure}[b]{0.5\textwidth}
\includegraphics[width=8.5 cm, height=5.5 cm]{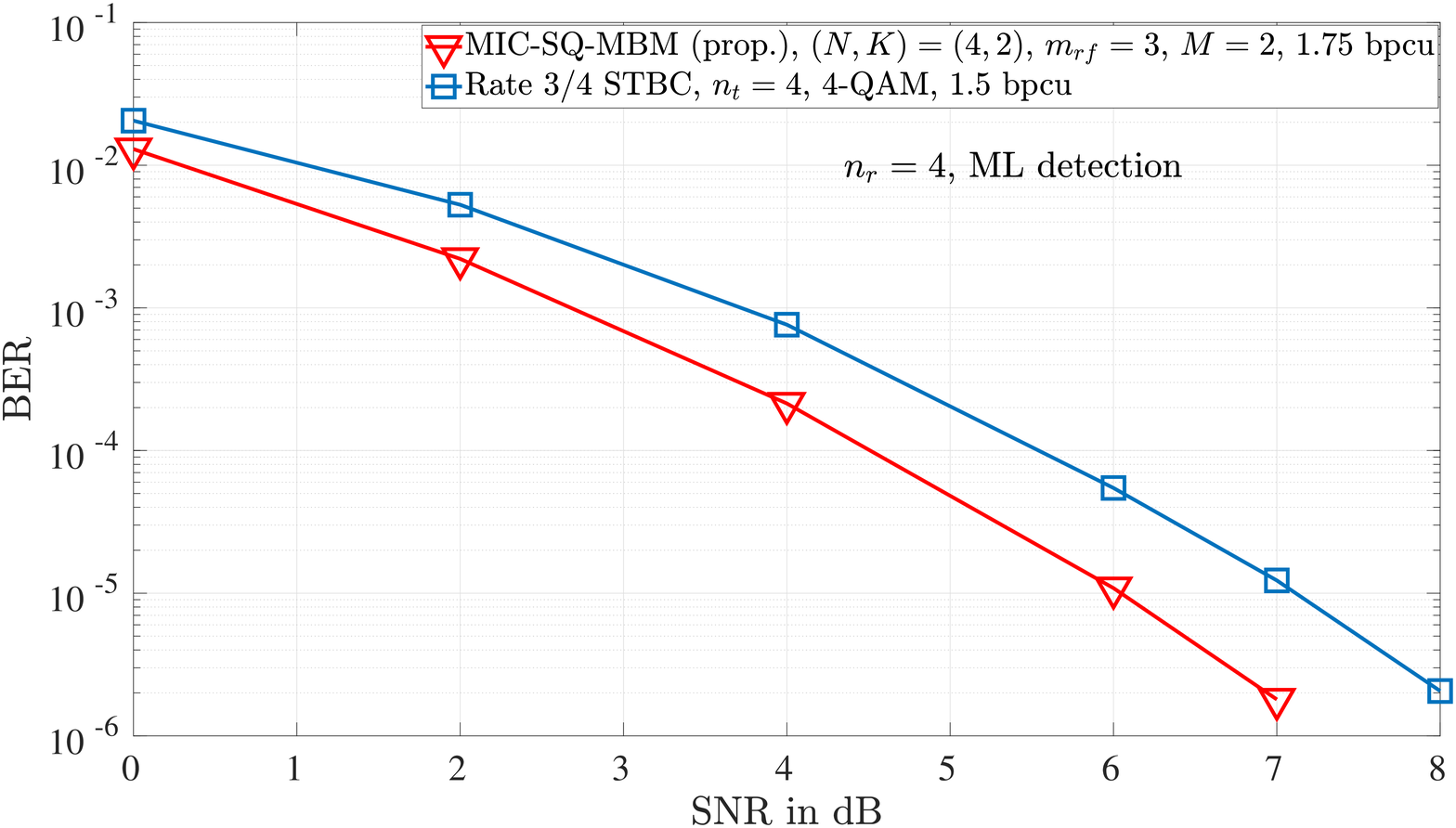}
\caption{}
\label{fig:STBC1}
\end{subfigure}%

\begin{subfigure}[b]{0.5\textwidth}
\includegraphics[width=8.5 cm, height=5.5 cm]{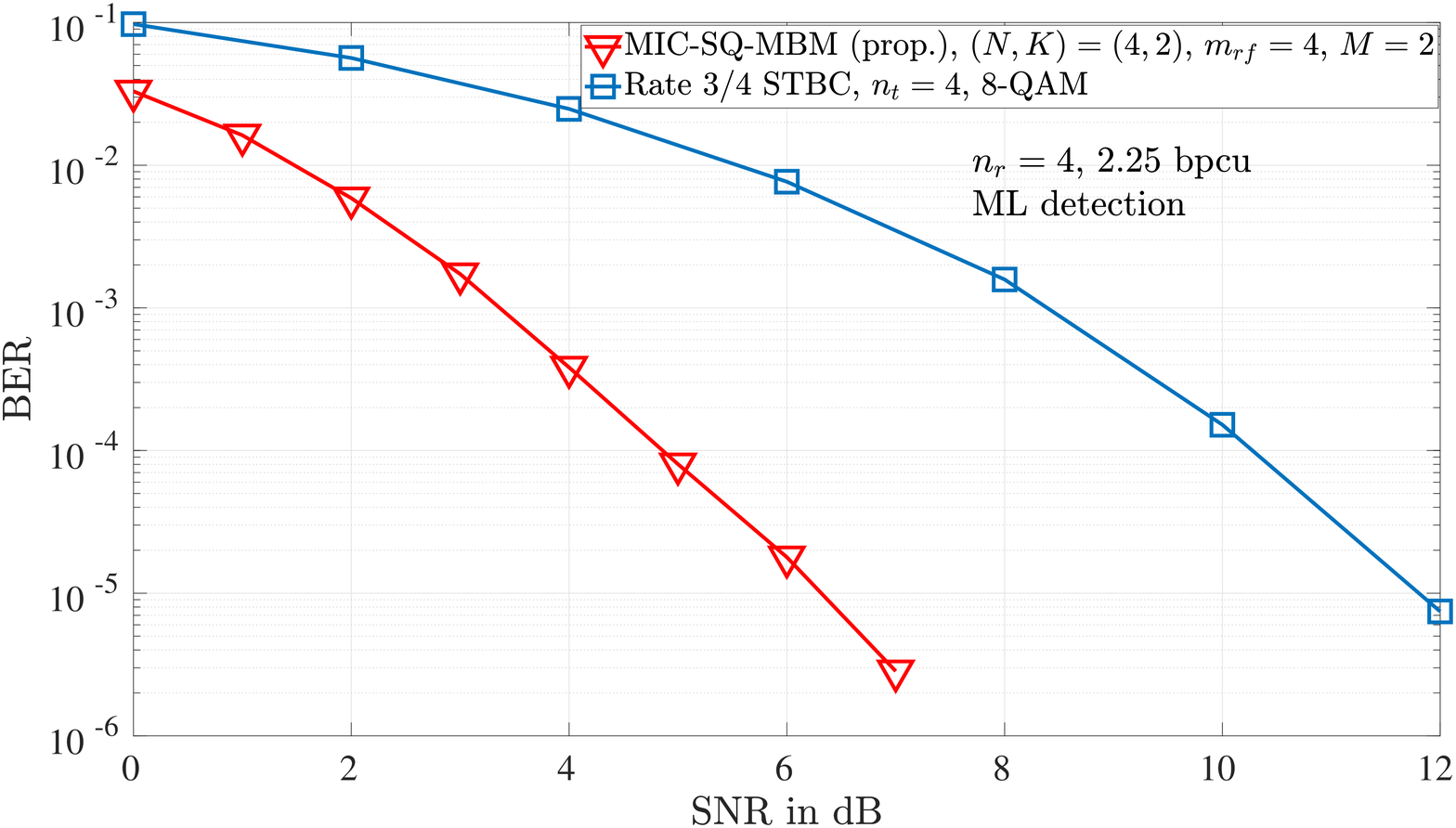}
\caption{}
\label{fig:STBC2}
\end{subfigure}%
\caption{{\color{black} BER performance comparison between MBM systems 
using the proposed constellation and STBC systems.}}
\label{fig:STBC_comparison}
\end{figure}

\section{Conclusions} 
\label{sec6}
We considered the problem of efficient constellation/signal set design 
for media-based modulation, and showed that block codes and squaring 
construction can be effectively used to design MBM signal sets with good 
distance properties. The proposed design approach was shown to result in 
MBM signal sets with good distance properties and bit error performance. 
Numerical and simulation results showed that the proposed MBM signal set 
can lead to significant advantages in terms of SNR and number of receive 
antennas compared to conventional MBM signal set. Through analysis and
validating simulations we established that the asymptotic diversity order
of the proposed signal set is the same as that of conventional MBM signal
set. However, an interesting observation is that in the low-to-medium
SNR regime, the proposed signal set was found to achieve a much higher 
diversity slope compared to that of conventional MBM signal set. This 
has resulted in significant SNR gains (e.g., 7 dB gain at $10^{-5}$ BER) 
compared to conventional MBM signal set. 
{\color{black} We note that the ML detection complexity grows 
exponentially with the block size $N$ and hence exhaustive search 
becomes infeasible for large block sizes. The structured sparsity in the
proposed constellation and the trellis structure of the non-zero symbols
obtained by squaring construction can be exploited to design low-complexity 
signal detection algorithms, which can be a potential topic future work.}


\begin{thebibliography}{99}

\bibitem{mbm1}
A. K. Khandani, ``Media-based modulation: a new approach to wireless
transmission,'' in {\em Proc. IEEE ISIT'2013}, Jul. 2013, pp. 3050-3054. 

\bibitem{mbm2}
A. K. Khandani, ``Media-based modulation: converting static Rayleigh
fading to AWGN,'' in {\em Proc. IEEE ISIT'2014}, Jun./Jul. 2014, 
pp. 1549-1553. 

{\color{black}
\bibitem{mbm2a}
E. Seifi, M. Atamanesh, and A. K. Khandani, ``Media-based MIMO: a new 
frontier in wireless communications,'' online: 
arXiv:1507.07516v3 [cs.IT] 7 Oct 2015.
}

\bibitem{mbm3}
E. Seifi, M. Atamanesh, and A. K. Khandani, ``Media-based MIMO: 
outperforming known limits in wireless,'' in {\em Proc. IEEE ICC'2016}, 
May 2016, pp. 1-7.

\bibitem{mbm4}
Y. Naresh and A. Chockalingam, ``On media-based modulation using RF 
mirrors,'' {\em IEEE Trans. Veh. Tech.}, vol. 66, no. 6, pp. 4967-4983, 
Jun. 2017.	


\bibitem{mbm6}
E. Basar and I. Altunbas, ``Space-time channel modulation,'' {\em IEEE 
Trans. Veh. Tech.}, vol. 66, no. 8, pp. 7609-7614, Aug. 2017.

\bibitem{mbm7}
I. Yildirim, E. Basar, and I. Altunbas, ``Quadrature channel modulation,'' 
{\em IEEE Wireless Comm. Lett.},  vol. 6, no. 6, pp. 790-793, Dec. 2017.

\bibitem{mbm8}
N. Pillay and H. Xu, ``Quadrature spatial media-based modulation with 
RF mirrors,'' {\em IET Communications}, vol. 11, no. 16, pp. 2440-2448, 
2017.

\bibitem{mbm9}
B. Shamasundar, S. Jacob, T. Lakshmi Narasimhan, and A. Chockalingam, 
``Media-based modulation for the uplink in massive MIMO systems,''
{\em IEEE Trans. Veh. Tech.}, vol. 67, no. 9, pp. 8169-8183, Sep. 2018. 

{\color{black}
\bibitem{vinoy1}
A. Roy and K. J. Vinoy, ``A reconfigurable screen in the antenna nearfield 
for media-based modulation scheme,'' {\em IEEE Intl. Microwave and 
RF Conf.,} 
Dec. 2018.


\bibitem{duman}
M. Hasan, I. Bahceci, M. A. Towfiq, T. M. Duman, and B. A. Cetiner, 
``Mode shift keying for reconfigurable MIMO antennas: performance 
analysis and antenna design,'' {\em IEEE Trans. Veh. Tech.,} vol. 68, 
no. 1, pp. 320-334, Jan. 2019.
}

\bibitem{constellation1}
C. R. Cahn, ``Combined digital phase and amplitude modulation
Communications,'' {\em IRE Trans. Commun. Syst.}, vol. CS-8,
pp. 150-154, 1960. 

\bibitem{constellation2}
C. N. Campopiano and B. G. Glazer, ``A coherent digital amplitude and 
phase modulation scheme,'' {\em IRE Trans. Commun. Syst.}, vol. CS-10, 
pp. 90-95, 1962.  

\bibitem{constellation3}
G. D. Forney, R. G. Gallager, G. R. Lang, F. M. Longstaff, and 
S. U. Qureshi, ``Efficient modulation for band-limited channels,'' 
{\em IEEE J. Sel. Areas Commun.}, vol. SAC-2, no. 5, pp. 632-647, 
Sep. 1984.

\bibitem{constellation4}
G. D. Forney, ``Multidimensional constellations - part I: introduction, 
figures of merit, and generalized cross constellations,'' {\em IEEE J. 
Sel. Areas Commun.}, vol. 7, no. 6, pp. 877-892, Aug. 1989.

\bibitem{constellation5}
G. D. Forney, ``Multidimensional constellations - part II: Voronoi 
constellations,'' {\em IEEE J. Sel. Areas Commun.}, vol. 7, no. 6, 
pp. 941-958, Aug. 1989.

\bibitem{constellation6}
L. F. Wei, ``Trellis-coded modulation with multidimensional 
constellations,'' {\em IEEE Trans. Inform. Theory}, vol. 33, 
pp. 483-501, 1987. 

\bibitem{coset1}
G. D. Forney, ``Coset codes - part I: introduction and geometrical 
classification,'' {\em IEEE Trans. Inform. Theory}, vol. 34, no. 5, 
pp. 1123-1151, Sep. 1988.

\bibitem{coset2}
G. D. Forney, ``Coset codes - part II: binary lattices and related 
codes,'' {\em IEEE Trans. Inform. Theory}, vol. 34, no. 5, pp. 1152-1187,
Sep. 1988.

\bibitem{lc1}
H. El Gamal, G. Caire, and M. O. Damen, ``On the optimality of lattice 
space-time (LAST) coding,'' in {\em Proc. IEEE ISIT'2004}, Jun. 2004,
pp. 98.

\bibitem{lc2}
K. Raj Kumar and G. Caire, ``Space-time codes from structured lattices,'' 
{\em IEEE Trans. Inform. Theory,} vol. 55, no. 2, pp. 547-556, Feb. 2009. 

\bibitem{lc3}
O. Shalvi, N. Sommer, and M. Feder, ``Signal codes: convolutional lattice 
codes,'' {\em IEEE Trans. Inform. Theory}, vol. 57, no. 8, pp. 5203-5226, 
Aug. 2011.

\bibitem{nonbinary}
R. A. Carrasco and M. Johnston,  {\em Non-binary error control coding for 
wireless communication and data storage}, John Wiley \& Sons, Nov. 2008.

\bibitem{stbc1}
V. Tarokh, N. Seshadri, and A. R. Calderbank, ``Space-time codes for 
high data rate wireless communication: performance criterion and code 
construction, {\em IEEE Trans. Inform. Theory}, vol. 44, no. 2, 
pp. 744-765, Mar. 1998.

\bibitem{stbc2}
J. Guey, M. Fitz, M. Bell, and W. Kuo, ``Signal design for transmitter 
diversity wireless communication systems over Rayleigh fading channels,'' 
{\em IEEE Trans. Commun.}, vol. 47, no. 4, pp. 527-537, Apr. 1999.  

\bibitem{micmbm}
B. Shamasundar, K. M. Krishnan, T. L. Narasimhan, and A. Chockalingam, 
``MAP-index coded media-based modulation,'' {\em IEEE Comm. Lett.}, 
vol. 22, no. 12, pp. 2455-2458, Dec. 2018.


\end{thebibliography}
\end{document}